\documentclass[12pt,draftclsnofoot,onecolumn]{IEEEtran}
\usepackage{amsmath,amsfonts,amsthm,amssymb,bbm}
\usepackage[margin=1in]{geometry}
\usepackage[normalem]{ulem}
\usepackage{mathrsfs}

\usepackage[hyphens]{url}
\usepackage{hyperref}
\usepackage{graphicx}
\usepackage{verbatim}
\usepackage{color}

\DeclareMathOperator*{\esssup}{ess\,sup}

\newtheorem{theorem}{Theorem}[section]
\newtheorem{assumption}{Assumption}[section]
\newtheorem{corollary}{Corollary}[section]
\newtheorem{lemma}{Lemma}[section]
\newtheorem{proposition}{Proposition}[section]

\newtheorem{definition}{Definition}[section]
\newtheorem{remark}{Remark}[section]
\newtheorem{example}{Example}[section]

\renewcommand{\P}{\mathbb{P}}

\newcommand{\R}{\mathbb{R}}
\newcommand{\E}{\mathbb{E}}

\newcommand{\N}{\mathbb{N}}

\newcommand{\cN}{\mathcal{N}}

\newcommand{\cD}{\mathcal{D}}
\newcommand{\cG}{\mathcal{G}}

\newcommand{\A}{\mathcal{A}}

\newcommand{\eps}{\varepsilon}

\newcommand{\Pqgv}{\mathsf{P}^{q,g}_{\nu}}

\newcommand{\Pginf}{\mathsf{P}^{q=0,g}_{\infty}}

\newcommand{\Eqv}{\mathbb{E}^{q}_{\nu}}
\newcommand{\Eq}{\mathbb{E}^{q,\emptyset}_0}
\newcommand{\Eqgv}{\mathbb{E}^{q,g}_{\nu}}
\newcommand{\Eqg}{\mathbb{E}^{q,g}_{0}}

\newcommand{\nada}[1]{}

\begin{document}
\title{Asymptotic Optimality in Byzantine Distributed Quickest Change Detection}

\author{Yu-Chih Huang, Yu-Jui Huang, and Shih-Chun Lin

\thanks{This paper was presented in part at the 2019 IEEE International Symposium on Information Theory \cite{ISIS19QCDconverse, ISIS19MQCD}.}

\thanks{Y.-C. Huang is with the Department of Communication Engineering, National Taipei University, 237 Sanxia District, New Taipei City, Taiwan (email: ychuang@mail.ntpu.edu.tw).}

\thanks{Y.-J. Huang is with the Department of Applied Mathematics, University of Colorado at Boulder, CO 80309 , USA (email: yujui.huang@colorado.edu)}

\thanks{S.-C. Lin is with the Department of Electronic and Computer Engineering, National Taiwan University of Science and Technology, 106 Daan District, Taipei City, Taiwan (email: sclin@ntust.edu.tw)}

\thanks{The authors are ordered alphabetically.}
}

\maketitle

\begin{abstract}
The Byzantine distributed quickest change detection (BDQCD) is studied, where a fusion center monitors the occurrence of an abrupt event through a bunch of distributed sensors that may be compromised. We first consider the binary hypothesis case where there is only one post-change hypothesis and prove a novel converse to the first-order asymptotic detection delay in the large mean time to a false alarm regime. This converse is tight in that it coincides with the currently best achievability shown by Fellouris \textit{et al.}; hence, the optimal asymptotic performance of binary BDQCD is characterized. An important implication of this result is that, even with compromised sensors, a 1-bit link between each sensor and the fusion center suffices to achieve asymptotic optimality. To accommodate multiple post-change hypotheses, we then formulate the multi-hypothesis BDQCD problem and again investigate the optimal first-order performance under different bandwidth constraints. A converse is first obtained by extending our converse from binary to multi-hypothesis BDQCD. Two families of stopping rules, namely the simultaneous $d$-th alarm and the multi-shot $d$-th alarm, are then proposed. Under sufficient link bandwidth, the simultaneous $d$-th alarm, with $d$ being set to the number of honest sensors, can achieve the asymptotic performance that coincides with the derived converse bound; hence, the asymptotically optimal performance of multi-hypothesis BDQCD is again characterized. Moreover, although being shown to be asymptotically optimal only for some special cases, the multi-shot $d$-th alarm is much more bandwidth-efficient and energy-efficient than the simultaneous $d$-th alarm. Built upon the above success in characterizing the asymptotic optimality of the BDQCD, a corresponding leader-follower Stackelberg game is formulated and its solution is found.
\end{abstract}

\section{Introduction}
The problem of quickest change detection (QCD), a.k.a. sequential change detection, studies detecting an abnormal event as quickly as possible after its occurrence at a deterministic but unknown time, subject to a certain false alarm rate. It has many applications and has been extensively researched since the early works \cite{page54, lorden71, moustakides86}. In these works, it is assumed that there is only one post-change hypothesis, which we refer to as the binary case. When there are multiple post-change hypotheses, the problem is referred to as multi-hypothesis QCD and has been investigated in \cite{nikiforov95,TV02}. A nice tutorial on QCD can be found in \cite{vvv_tutorial}. However, recent applications of cyber-physical systems (CPS) \cite{wurm2016introduction} typically involve multiple distributed sensors monitoring the event and reporting their observations to the fusion center via bandwidth-limited links. For example, the abnormal changes of voltage waveforms in smart grids are harmful to delicate electronic devices and recent advances of massive machine-type communications (mMTC) or internet of things (IoT) \cite{IoT17} allow the usage of advanced cyber-physical infrastructures for monitoring voltage quality events distributively \cite{AMI_PQ}. Moreover, some sensors, whose identities are unknown to the fusion center, may be compromised and may try to sabotage the detection task. Motivated by these applications, this paper considers the decentralized version of QCD, where a fusion center monitors the event through distributed sensors, with compromised sensors collaboratively forming attack. This problem has been studied in \cite{bayraktar15, BFL16} and is called Byzantine distributed QCD (BDQCD).

In \cite{bayraktar15}, a special case of binary BDQCD, only one compromised sensor is considered. A decision rule called second-alarm rule, where the fusion center declares the occurrence of the event once it receives the second local report from sensors, is proposed and its asymptotic performance is analyzed. In \cite{BFL16}, the general binary BDQCD problem with infinite-bandwidth links and that with 1-bit links are investigated. Multiple rules are proposed and their corresponding asymptotic performance are analyzed. In the presence of infinite-bandwidth links, the low-sum CUSUM scheme proposed in \cite{BFL16} achieves the best asymptotic performance among the schemes in \cite{BFL16}. Among the rules with 1-bit links proposed in \cite{BFL16}, the voting rule that declares the occurrence of the event after the number of received local reports exceeds a certain threshold has the best first-order asymptotic performance. When the threshold is set to be the total number of honest sensors, the asymptotic performance of the voting rule, called the consensus rule in this special case, reaches its maximum and also attains the best asymptotic performance in \cite{BFL16}. Two questions naturally arise from this premise:
\begin{enumerate}
  \item What is the fundamental limit of the first-order asymptotic performance of binary BDQCD?
  \item How to handle this problem when there are multiple post-change hypotheses?
\end{enumerate}

For the first question, despite the exciting results in \cite{bayraktar15,BFL16}, it is thus far unclear what the best first-order asymptotic performance of binary BDQCD is. Although one non-trivial converse can be easily obtained by assuming that a genie reveals to the fusion center the identities of all honest sensors (this simple converse will be presented in Section~\ref{sec:prior}), this converse bound and the best achievable asymptotic performance in \cite{BFL16} do not match. This indicates that either the best achievable scheme thus far is not optimal or the converse is not tight, or both.

The first contribution of this work is to prove a new converse to the first-order asymptotic performance for binary BDQCD. In the proof, we first construct an attack strategy for the compromised sensors and then construct a genie who reveals just enough information to the fusion center. After that, inspired by the proof technique in \cite{IHsiang18}, we transform the original problem into a centralized QCD problem. Last, evaluating the corresponding optimal CUSUM procedure establishes the new converse. The converse turns out to coincide with the best achievable first-order asymptotic performance known to date; thereby, the fundamental limit of the first-order asymptotic performance of binary BDQCD is characterized. Specifically, our converse confirms that both the consensus rule (using only 1-bit links) and the low-sum-CUSUM rule (using infinite-bit links) in \cite{BFL16} achieve the optimal first-order scaling. The first optimality unveils an important implication that, at least asymptotically, {\it 1-bit links suffice even with compromised sensors}. As a byproduct of our proof, we explicitly construct an attack strategy, called the reverse attack, where each compromised sensor generates fake i.i.d. observations according to post-change and pre-change distributions before and after the change time (i.e., with pre-change and post-change distributions swapped), and form local reports based on these fake observations. In spite of abandoning potential cooperation among compromised sensors, this reverse attack turns out to be strong enough for us to prove a tight lower bound on the asymptotic performance of BDQCD\footnote{Throughout the paper, such an attack strategy is said to be an asymptotically worst (to the fusion center) attack, or simply a worst attack.}. We note that although our converse is inspired from \cite{IHsiang18}, the one-shot hypothesis testing problem studied in \cite{IHsiang18} is fundamentally different to our sequential change detection that deals with observation sequences having an unknown change time. More detailed comparisons with \cite{IHsiang18} are provided at the end of Section \ref{subsec:proof_converse_binary}.

Our second contribution is to tackle the second question listed above and extend the framework of BDQCD to the multi-hypothesis setting. To this end, we formulate the multi-hypothesis version of the BDQCD problem. We then demonstrate that blindly adopting the existing procedure in \cite{TV02}\cite{BFL16}\cite{BF16} may result in catastrophic events. Two novel stopping rules are proposed and analyzed, which requires $\log(Q)$-bit and $Q$-bit noiseless links, respectively, for BDQCD with $Q+1$ hypotheses. In our delay analysis, we prove an asymptotic dominance result, which confirms the intuition that although there are multiple hypotheses, for each one being considered, we only have to examine the statistics of another hypothesis that is closest (in the sense of Kullback–Leibler (KL) divergence) to the hypothesis being considered. The converse for the binary case is also extended to the multi-hypothesis setting and it is shown that proposed stopping rules can achieve the optimal first-order asymptotic performance under different bandwidth constraints; therefore, the asymptotically optimal performance of multi-hypothesis BDQCD is again characterized.

Last but not least, we formulate a leader-follower Stackelberg game \cite{basar1999dynamic} where the fusion center and honest sensors form the leader while the compromised sensors form the follower. The first-order optimality mentioned above yields the game solution, where the leader adopts the aforementioned asymptotically optimal stopping rule and the follower employs the corresponding worst attack.

\subsection{Organization}
The rest of the paper is organized as follows. We will separately introduce the problem of binary BDQCD and multi-hypothesis BDQCD in Sections~\ref{subsec:B_BDQCD} and \ref{subsec:M_BDQCD}, respectively, and review the current state-of-the-art in Section~\ref{sec:prior}. We will then split our discussion into two parts, namely the binary BDQCD in Section~\ref{sec:B_BDQCD} and multi-hypothesis BDQCD in Section \ref{sec:mul_BDQCD}, even though the former is a special case of the latter. The main reasons are: 1) binary BDQCD is of substantial interest in its own right and has been a subject of research in the literature \cite{bayraktar15,BFL16}, 2) due to the nature of having multiple post-change hypotheses, one has to consider a sequence of stopping times, which is in sharp contrast to binary BDQCD where only one stopping time is considered, 3) our proofs for tight converse bounds in Sections~\ref{sec:B_BDQCD} and \ref{sec:mul_BDQCD} are quite involved with heavy notations and it is better to start with the binary case in Section~\ref{sec:B_BDQCD} and mention only notable differences later in Section~\ref{sec:mul_BDQCD}, and 4) for achievability, we borrow existing results in \cite{BFL16} for the binary BDQCD, while we devise new efficient stopping rules for the multi-hypothesis BDQCD whose first-order asymptotic performance coincides with our converse bound. Finally, a leader-follower Stackelberg game and its solution are then presented in Section~\ref{sec:game}.

\subsection{Notational conventions}
For a positive integer $K$, define $[K]:=\{1,\ldots,K\}$ and $[K]^+ = \{0\}\cup[K]$. Function $(x)^+$ outputs $x$ if $x \geq 0$ and zero otherwise. For two real functions $f_1(x)$ and $f_2(x)$, as $x\rightarrow \infty$, we write $f_1(x)\sim f_2(x)$ when $f_1(x)/f_2(x)\rightarrow 1$ and $f_1(x)\gtrsim f_2(x)$ when $\liminf (f_1(x)/f_2(x)) \geq 1$. The $o(.)$ and $w(.)$ follow the asymptotic notations in \cite{cormen2009introduction}. 

\section{Problem formulation}\label{sec:problem}
In this section, we formally state the problem of BDQCD. We will first describe binary BDQCD and then formulate the generalization to the multi-hypothesis case.

\subsection{Binary BDQCD}\label{subsec:B_BDQCD}
The binary BDQCD problem consists of a fusion center and $K$ sensors indexed by $[K]$. Among these sensors, there is an \textit{unknown} subset $\cN\subset [K]$ of honest sensors, with the remaining $M:= K-|\cN|$ sensors being potentially compromised. The goal of the honest sensors is to monitor an event and help the fusion center decide whether the event has occurred, while the goal of compromised sensors is to collaboratively confuse the fusion center. Although the exact information about which sensors are honest and which sensors are compromised is unknown, we assume that $M$, the maximum number of sensors the attacker can compromise, is known by the fusion center. Moreover, it is assumed that there are more honest sensors than compromised sensors, i.e., $|\cN| > M$. The observations of all $K$ sensors are sequences of independent random variables with known distributions, subject to the same distribution change at an unknown but deterministic time $\nu$. 
Before the change time $\nu$, sensor $k$'s observations $X^k_1, X^k_2, \ldots, X^k_{\nu}$ are independent and identically distributed (i.i.d.) with the density $P_0$, while $X^k_{\nu+1},X^k_{\nu+2}, \ldots$ are i.i.d. with the density $P_1$. If the change never happens, i.e., $\nu=\infty$, $X^k_t$ are i.i.d. with $P_0$ for all $t$. We denote by $\mathbf{X}_t=[X^1_t,X^2_t,\ldots,X^K_t]$ the collection of observations at time $t$ and we use the notation $\mathbf{X}_{t_1}^{t_2}$ for $t_1<t_2$ to denote the collection $[\mathbf{X}_{t_1},\mathbf{X}_{t_1+1},\ldots,\mathbf{X}_{t_2}]$. Also, we define the KL divergence from $P_0$ to $P_1$ as \cite{CT91} to be $I := \int \log\left(\frac{P_1(x)}{P_0(x)}\right) P_1(x) dx$. Throughout the paper, we assume that $I$ is finite and strictly positive and
\begin{equation}
    \int \log\left(P_1(x)/P_0(x)\right)^2 P_1(x) dx < \infty.
\end{equation}

All the local reports from honest or compromised sensors belong to the set $\mathcal{X}$, which satisfies the underlying bandwidth constraint on the noiseless link between each sensor and the fusion center. It is worth emphasizing that this setting encompasses many scenarios discussed in existing works including $\mathcal{X}=\{0,1\}$ and $\mathcal{X}=\mathbb{R}$ in \cite{BFL16} and $\mathcal{X}$ being a set of finite alphabets in \cite{bounds}. At each time index $t$, the honest sensor $k$ individually makes a local decision by mapping its own observations up to time $t$ to an element in $\mathcal{X}$, and then chooses to report it or not according to the adopted reporting mechanism. Based on the received local reports from all sensors, the fusion center adopts a stopping rule to determine when to declare that the event has occurred. A change detection rule includes such a stopping rule and local rules at honest sensors. The $M$ compromised sensors, on the other hand, try to disrupt/confuse the fusion center by sending attack signals in $\mathcal{X}$. We assume a very powerful attacker that knows the exact change-time $\nu$ and has access to the current and past observations of all nodes. The symbols sent by the compromised sensors at time $t$ are then produced by $g$, a function (called an attack strategy) with inputs $\nu$, $\mathbf{X}_{1}^{t}$, and the change detection rule. We denote by $\cG$ the set of all attack strategies including all possible $g$ with no more than $M$ compromised sensors. Following \cite{BFL16}, we analyze the performance of a rule by its worst-case expected detection delay and mean time to a false alarm in the sense of Lorden \cite{lorden71}, under the worst attack strategy among $\cG$. Specifically, let $T$ be the stopping time of a rule, we define the performance metrics as follows.
\begin{itemize}
  \item \textbf{Detection Delay:} The worst-case mean detection delay
\begin{equation}\label{eq:delay_binary}
    \cD[T]:= \sup_{g\in\cG,\nu} \esssup \E^{g}_{\nu}[ (T-\nu)^+ |\mathbf{X}_1^{\nu}],
\end{equation}
where $\E^{g}_{\nu}[.]$ means the expectation is taken w.r.t. $P_0$ when $t\leq\nu$ and w.r.t. $P_1$ when $t>\nu$ under the attack strategy $g\in\cG$.

  \item \textbf{False Alarm:} Without any abnormal changes (i.e. $\nu=\infty$), the worst-case mean time to a false alarm is
\begin{equation}\label{eq:falseAlarmTime_multiple}
  \A[T]:= \inf_{g\in\cG} \E^{g}_{\infty}[ T],
\end{equation}
where $\E^{g}_{\infty}[.]$ means the expectation is w.r.t. $P_0$ for all $t$ (i.e., $\nu=\infty$) under the attack strategy $g\in\cG$.

\end{itemize}
The main theme of this paper is to investigate the optimal asymptotic behavior of how the expected detection delay scales with the mean time to a false alarm in the worst case. Specifically, for an optimal BDQCD rule with stopping time $T$ satisfying $\A[T] \geq \gamma$, we want to characterize how $\cD[T]$ grows with $\gamma$ as $\gamma\rightarrow\infty$.

\begin{remark}
    We would like to emphasize that our setting is slightly different from that in \cite{BFL16}. In \cite{BFL16}, among honest sensors, some are affected and some are unaffected by the change. Similar scenarios with no Byzantine attack are also considered in the literature, with exactly one unknown sensor \cite{Tarta05,tarta_vvv_book04} or a subset of all sensors \cite{Mei10,XieSiegmund13} being affected. For the unaffected, their observations are sampled i.i.d. from $P_0$ even after the change. In this paper, we do not consider unaffected sensors purely to avoid heavy notation. With slight modifications, our results can be easily extended to include unaffected sensors. This statement remains true even for the multi-hypothesis setting discussed later.
\end{remark}

\subsection{Multi-Hypothesis BDQCD}\label{subsec:M_BDQCD}
For the multi-hypothesis version of BDQCD, we again consider a network with a fusion center and $K$ distributed sensors. There is an \textit{unknown} subset $\cN\subset [K]$ of honest sensors, with the remaining $M:= K-|\cN|$ sensors being compromised. The fusion center tries to monitor an abrupt event and decide whether the event has occurred regardless of which type it is\footnote{This is aligned with \cite{nikiforov95, TV02} . In many applications, once a change has been detected, the operator can respond to it quickly and find out which type it is.}. The observations of all $K$ sensors are sequences of independent random variables with known distributions, subject to the same distribution change at an unknown but deterministic time $\nu$. After this distribution change, there are $Q$ different possible types. Specifically, let $P_0$ be the pre-change probability density function (PDF) and $P_1, \ldots, P_Q$ the post-change PDF corresponding to the states $1,\ldots,Q$, respectively. For each $k\in [K]$, we denote by $X^k_t$ the observation made by sensor $k$ at time $t$. We can now define $Q+1$ different hypotheses as follows. Under the hypothesis $H_q$, $q \in [Q]$, the random variables $X^k_1, X^k_2, \ldots, X^k_{\nu}$ are i.i.d. with the PDF $P_0$, while $X^k_{\nu+1},X^k_{\nu+2}, \ldots$ are i.i.d. with the PDF $P_q$. Under the hypothesis $H_0$, $X^k_t$ are i.i.d. with the PDF $P_0$ for all $t$. We write $\mathbf{X}_t=[X^1_t,\ldots,X^K_t]$ for each $t$ and denote by $\mathbf{X}_{t_1}^{t_2}$ the collection of $\mathbf{X}_{t_1},\mathbf{X}_{t_1+1},\ldots, \mathbf{X}_{t_2}$ for each $t_1,t_2$ with $t_2>t_1$.

As the binary case, there is a noiseless link of a finite or infinite number of bits associated with each sensor to the fusion center. At each time $t$, an honest sensor $k$ makes a local decision individually by mapping its own observations up to time $t$ to an element in $\mathcal{X}$ satisfying the bandwidth constraint. Depending on the adopted reporting mechanism and the bandwidth constraint, each sensor decides whether it should alarm the fusion center through the channel it is associated with. 
The $M$ compromised sensors, on the other hand, try to disrupt/confuse the final decision of fusion center by sending attack signals which again belong to $\mathcal{X}$.

As \cite{nikiforov95, TV02}, let us define the sequence of alarm times $0=T_0<T_1<T_2<\ldots<T_{\rho}<\ldots,$ where $T_{\rho}$ is the alarm time using sensor observations after previous alarm time $T_{\rho-1}$, that is, $\mathbf{X}_{T_{\rho-1}+1}, \mathbf{X}_{T_{\rho-1}+2}, \ldots$; then the stopping time for type $q\in[Q]$ is defined as
\begin{equation}\label{eq:stopping_rule_multi}
    T^{q} = \inf_{\rho\geq 1}\{T_{\rho}:\hat{q}_{\rho}=q\},
\end{equation}
where $\hat{q}_{\rho}$ is the decision at the fusion center declared at time $T_{\rho}$. Here, we use the convention $\inf\{\emptyset\}=\infty$ and it is possible that $T^{q}=\infty$, which corresponds to the case when the fusion center never declares change of type $q$. With a little abuse of notation, when only the first alarm time $T_1$ of a rule $T$ matters, we sometimes simply write $T_1$ as $T$. Let $g$ be an attack strategy of the $M$ compromised sensors. We assume that the attacker knows $\nu$, $\mathbf{X}_1^t$, and the global decision rule (including both the stopping rule at the fusion center and local rule at each sensor), and hence $g$ is a function of these arguments. We also write $g=\emptyset$ when all the compromised sensors are absent. When a change under hypothesis $H_q$, $q\in[Q]$, happens at time $\nu$ and the strategy employed by the $M$ compromised sensors is $g$, the underlying probability measure is denoted by $\Pqgv$. Moreover, when no change ever happens, i.e., $\nu=\infty$, we denote by $\Pginf$ the underlying probability measure.

Following the single-sensor case \cite{nikiforov95, TV02}, we define the performance metrics as follows:
\begin{itemize}
  \item \textbf{Detection Delay:} The worst-case mean detection delay is given by
\begin{equation}\label{eq:ess_delay_multiple}
    \cD[T] :=\sup_{q \in [Q]} \sup_{g,\nu}\esssup \Eqgv[ (T-\nu)^+|\mathbf{X}_1^{\nu}],
\end{equation}

  \item \textcolor{black}{\textbf{False Alarm or False Isolation:} The worst-case mean time to a false alarm or a false isolation is given by
\begin{equation}\label{eq:falseIsolationTime_multiple}
  \A[T] := \inf_{q \in [Q]^+} \inf_{g} \inf_{\hat{q}\in[Q]\setminus\{q\}}\Eqg[ T^{\hat{q}} ].
\end{equation}}
\end{itemize}
Our objective is again to design fault-tolerant decision rules such that $\cD[T]$ can be minimized, with large $\A[T] \geq \gamma$.

\begin{remark}
The detection delay defined in \cite{nikiforov95, TV02} for the scenario with one single honest sensor is given by
\begin{equation}
    \sup_{q \in [Q]} \sup_{\nu}\esssup \Eqv[ T-\nu|T>\nu,X^1_1,\ldots,X^1_{\nu}],
\end{equation}
where $X^1_1,\ldots,X^1_{\nu}$ are the observations of the (single and honest) first sensor up to time $\nu$. One may notice that \eqref{eq:ess_delay_multiple} is not quite of the same form as the above definition. In Lemma~\ref{lma:equivalence} in Appendix~\ref{apx:lemma}, we prove that even with multiple honest and compromised sensors, these two forms are equivalent and one is free to work with either of them.
\end{remark}

Before leaving this section, we quickly review the asymptotically optimal matrix CUSUM algorithm in \cite{TV02} when there is just a single honest sensor $|\cN|=1, M=0,$ and $Q \geq 1$. For each hypothesis $q\in[Q]$, this honest sensor (with index $k=1$)  computes the CUSUM statistics $Y^k_t(q,j)$ for every $j\neq q\in[Q]^+$ at time $t$, recursively through $Y^k_0(q,j)=0$ and
\begin{equation}\label{eq_CUSUM_element}
Y^k_t(q,j)=\left(Y^k_{t-1}(q,j) + \ell^k_t(q,j)\right)^+,
\end{equation}
where $\ell^k_t(q,j)= \log \frac{P_{q}(X^k_t)}{P_{j} (X^k_t)}$ is the log-likelihood ratio (LLR) between $P_q$ and $P_j$. The results are put into a $Q\times Q$ matrix $\mathbf{Y}_t$ with the $q$th row given by
\begin{eqnarray} \label{eq_matrix_CUSUM_matirx}
  \mathbf{Y}^k_t:=[Y^k_t(q,0),\cdots,Y^k_t(q,j),\cdots,Y^k_t(q,Q)].
\end{eqnarray}
Let $Y^k_{t,q} = \min_{j\in [Q]^+, \; j \neq q} Y^k_t(q,j)$ be the minimum of the $q$th row. The matrix CUSUM procedure in \cite{TV02} locally determines that the event has occurred at the first time that any $Y^k_{t,q}, q \in [Q]$ exceeds a pre-defined threshold $h$. A hard decision is then alarmed, which means that the procedure terminates after this alarm and no other decisions will be further made.

\section{Prior work}\label{sec:prior}
In this section, we review some prior results directly relevant to the present work. We again split our discussion into the binary and multi-hypothesis cases.

\subsection{Binary BDQCD}\label{subsec:binary_prior}
For the considered binary BDQCD problem with $|\cN|=1$ and $M=0$ (i.e., no compromised sensor and thereby no Byzantine attack), the problem reduces to the standard QCD problem for which it was shown in \cite{lorden71,moustakides86} that Page's CUSUM procedure $T_{\textrm{single}}$ \cite{page54} achieves the optimal scaling that for $\A[T_{\textrm{single}}]=\gamma$, the expected detection delay scales like $\cD[T_{\textrm{single}}]\sim \log(\gamma) / I$ as $\gamma\rightarrow \infty$. For $M=0$ and general $|\cN|$, Mei in \cite{bounds} developed a scheme $T_{\textrm{consensus}}$, called the consensus rule, where each sensor performs CUSUM according to its local observations and sends a binary report to the fusion center, which declares the occurrence of the event when all the $|\cN|$ sensors simultaneously say so. It was then shown in \cite{bounds} that this scheme is asymptotically optimal that under $\A[T_{\textrm{consensus}}]=\gamma$, the expected detection delay scales like
\begin{equation}
    \cD[T_{\textrm{consensus}}]\sim \frac{\log(\gamma)}{|\cN|I}, \quad\text{as $\gamma\rightarrow \infty$.}
\end{equation}

In \cite{BF16}, Banerjee and Fellouris proposed two families of stopping rules for the same $M=0$ and general $|\cN|$ case. In the first family of stopping rules, which we refer to as the one-shot $d$-th alarm, each sensor performs the CUSUM procedure locally and only reports an alarm once at the first time the local CUSUM statistic exceeds a predefined threshold; the fusion center then stops and declares the event as soon as receiving $d\le |\cN|$ reports. In the second family of stopping rules, which is referred to as the $d$-voting rule, each sensor again performs the CUSUM procedure locally but gets to report multiple times whenever the local CUSUM statistic exceeds the threshold. The fusion center then stops and declares the event as soon as receiving $d\le|\cN|$ reports \textit{simultaneously}. The authors of \cite{BF16} analyzed the second-order asymptotic performance and the results revealed that even though it was shown in \cite{bounds} that the $d$-voting rule with $d=|\cN|$ (i.e., the consensus rule) achieves the first-order asymptotic performance, it might be better in practice for it to wait for only the majority of sensors' reports, i.e., setting $d=\lceil (|\cN|+1)/2 \rceil$.

Very recently, in \cite{BFL16}, binary BDQCD with general $|\cN|$ and $M$ was discussed and multiple schemes were analyzed. Among these schemes, the $d$-voting rule $\tau_{(d)}$, achieves the best scaling when $d=|\cN|$ is chosen\footnote{This is also called the consensus rule in \cite{BFL16}. But it is noted that here, we only wait for $|\cN|$, the number of honest sensors, responses rather than all $K$ responses.}. Specifically, it was shown in \cite{BFL16} that the following asymptotic performance can be achieved: 
\begin{theorem}[{\cite[Theorem 26]{BFL16}}]\label{thm:achieve}
    Let $d$ be an integer satisfying $M < d\leq |\cN|$. For $\A[\tau_{(d)}]=\gamma$, the worst-case mean detection delay of the $d$-voting rule scales like
    \begin{equation}\label{eqn:opt_scale}
        \cD[\tau_{(d)}]\sim \frac{\log\gamma}{(d-M)I}, \quad\text{ as $\gamma\rightarrow \infty$.}
    \end{equation}
\end{theorem}

The best asymptotic performance reported in \cite{BFL16} is the above one with $d=|\cN|$ (i.e., the consensus rule), which also coincides with another scheme in \cite{BFL16}, low-sum-CUSUM that requires infinite bandwidth. This leads us to conjecture that \eqref{eqn:opt_scale}, with $d=|\cN|$, is the optimal first-order behavior. A tight converse is then necessary to verify this conjecture.

We would like to point out that a non-trivial converse can be obtained by revealing the identities of all $|\cN|$ honest sensors and using the asymptotic optimality in \cite{bounds}, as detailed below.
\begin{theorem}[Simple converse]\label{thm:simple_converse}
    For any binary BCQCD rule $T$, with $\A[T] \geq \gamma$, the worst-case mean detection delay meets
    \begin{equation}
        \mathcal{D}[T]\gtrsim \frac{\log\gamma}{|\cN|I}, \quad\text{as $\gamma\rightarrow\infty$.}
    \end{equation}
\end{theorem}
\noindent Unfortunately, this converse is not tight compared to \eqref{eqn:opt_scale}.

\subsection{Multi-hypothesis BDQCD}\label{subsec:multi_prior}
To the best of our knowledge, the present work is the first to formulate and study the multi-hypothesis BDQCD. Prior to this work, the single honest sensor QCD problem with multiple hypothesis was first investigated in \cite{nikiforov95}, in which Nikiforov extended Lorden's framework to include multiple post-change hypotheses. Nikiforov in \cite{nikiforov95} also proposed an algorithm based on the concept of generalized likelihood ratio and showed the asymptotic optimality of this algorithm. In \cite{TV02}, by cleverly switching the order of $\max$ and $\min$ in the algorithm in \cite{nikiforov95}, Oskiper and Poor developed the matrix CUSUM algorithm that admits a recursive formula and hence can be efficiently implemented. Moreover, it was shown that, in addition to its low complexity, the matrix CUSUM procedure is also asymptotically optimal.   

\section{Binary BDQCD}\label{sec:B_BDQCD}
We consider the problem of binary DBQCD in this section. We first present the main result, that is, a tight converse to the first-order asymptotic performance of the worst-case detection delay, in Section~\ref{subsec:result_binary}. The proof of the main result is then given in Section~\ref{subsec:proof_converse_binary}.

\subsection{Main results of this section}\label{subsec:result_binary}
Here, we present the main result of this section, which is a new converse of the first-order asymptotic performance of BDQCD.
\begin{theorem}[Tight converse]\label{thm:new_converse}
    For any binary BDQCD rule $T$, with $\A[T] \geq \gamma$, the worst-case mean detection delay is lower bounded as
    \begin{equation} \label{eq:new_converse}
        \mathcal{D}[T]\gtrsim \frac{\log\gamma}{(|\cN|-M)I}, \quad\text{as $\gamma\rightarrow\infty$}.
    \end{equation}
\end{theorem}

A sketch of the proof of the new converse is outlined in Fig.~\ref{fig:proof} and the details are given in the next subsection. To prove this theorem, we first note that if the optimal asymptotic scaling is lower bounded by $\eta(\gamma)$ under an attack strategy, then it is also lower bounded by $\eta(\gamma)$ under the {\it worst} attack. We thus proceed by constructing an attack strategy in Section~\ref{sec:rev_attack}, called the {\it reverse attack}, which is later shown to be an asymptotically worst attack. We then, in Section~\ref{sec:genie}, construct a genie providing the identities of $|\cN|-M$ out of $|\cN|$ honest sensors and the local observations used for generating the local report at every sensor. After that, by absorbing the impact of the reverse attack into pre/post-change distributions, the problem is transformed into an equivalent centralized QCD problem in Section~\ref{sec:tQCD} for which CUSUM is known to be optimal. Finally, in Section~\ref{sec:est_converse}, evaluating the CUSUM procedure for the transformed problem reveals the connection to another non-Byzantine QCD with only $|\cN|-M$ honest sensors.
\begin{figure}
    \centering
    \includegraphics[width=3in]{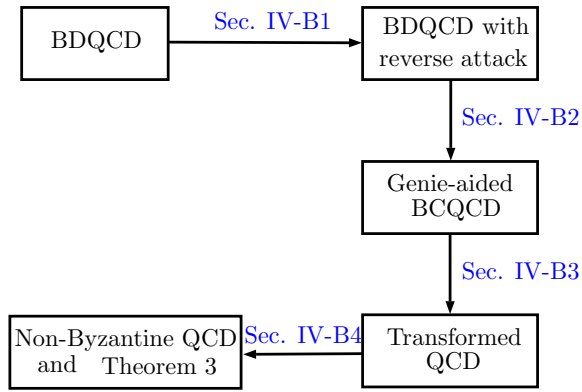}
    \caption{Diagram of proof steps.}
    \label{fig:proof}
\end{figure}

When comparing the main result of this section presented above and the achieviability result Theorem~\ref{thm:achieve}, one immediately characterizes the optimal first-order behavior of binary BDQCD as follows.
\begin{corollary}
    For an optimal binary BDQCD rule $T^*$, subject to $\A[T^*] \geq \gamma$, the first-order asymptotic worst-case mean detection delay is given by
    \begin{equation}
        \mathcal{D}[T^*] \sim \frac{\log\gamma}{(|\cN|-M)I},\quad\text{as $\gamma\rightarrow\infty$.}
    \end{equation}
\end{corollary}

\begin{remark}
    Supposed that, as in \cite{BFL16}, there is a subset $\mathcal{B}\subseteq\cN$ such that only sensors in $\mathcal{B}$ are affected by the change and those in $\cN\setminus\mathcal{B}$ have observations drawn i.i.d. according to $P_0$ all the time. We can slightly alter our genie in our proof so that it reveals the identities of $|\mathcal{B}|-M$ affected sensors and  the $|\cN\setminus\mathcal{B}|$ unaffected sensors. One can then follow the same technique to prove the following converse,
    \begin{equation}
        \mathcal{D}[T]\gtrsim \frac{\log\gamma}{(|\mathcal{B}|-M)I},\quad\text{as $\gamma\rightarrow\infty$.}
    \end{equation}
    Moreover, setting $d=|\mathcal{B}|$ in the $d$-voting rule achieves the above first-order scaling; hence, the optimal first-order asymptotic performance of this setting is also characterized as
    \begin{equation}
        \mathcal{D}[T^*]\sim \frac{\log\gamma}{(|\mathcal{B}|-M)I},\quad\text{as $\gamma\rightarrow\infty$.}
    \end{equation}
\end{remark}

\subsection{Proof of the converse for binary BDQCD}\label{subsec:proof_converse_binary}
The proof presented in this section follows closely the steps shown in Fig.~\ref{fig:proof}.

\subsubsection{The reverse attack}\label{sec:rev_attack}
For the ease of presentation in this proof, we define $P_{0,1}=P_0$ and $P_{1,1}=P_1$. Recall that each honest sensor $k$'s observation sequence $X_t^k$ is drawn i.i.d. according to $P_{0,1}$ before the change time $\nu$ and i.i.d. according to $P_{1,1}$ after $\nu$. We construct an attack strategy as follows. For each compromised sensors $k'$, it generates a fake observation sequence $X_t^{k'}$, which is then input to the assigned local decision function for forming the fake report. The fake observation sequence is generated i.i.d. according to $P_{0,2}$ and $P_{1,2}$ before and after the change time $\nu$, respectively. That is, the compromised sensors form fake reports according to observations based on wrong distributions. To establish the tight converse, we will set $P_{0,2}=P_{1,1}=P_1$ and $P_{1,2}=P_{0,1}=P_0$ in the very end of the proof; therefore, we call this attack strategy the ``reverse attack". However, most of the steps in the proof stay valid for general densities $P_{0,2}$ and $P_{1,2}$. Next, we will show that under this reverse attack, for any detection rule with mean time to a false alarm no less than $\gamma$, the mean detection delay is lower-bounded by the RHS of \eqref{eq:new_converse}. Note that by definition, the worst case delay in \eqref{eq:delay_binary} will also be lower-bounded by \eqref{eq:new_converse} automatically.

\subsubsection{Genie-aided Byzantine centralized QCD}\label{sec:genie}
First note that the worst case happens when there are $M$ compromised sensors. Also since the identities of the sensors are unknown, the fusion center cannot enhance the worst-case performance by selectively accepting reports. If the fusion center accepts reports from $K-K'$, $K'\leq |\cN|$, sensors only, in the worst case, the problem reduces to the BDQCD with $M$ compromised sensors and $|\cN|-K'$ honest sensors, which results in a worse performance. Moreover, when $K'>|\cN|$, we are left with only compromised sensors in the worst case, which is obviously worse than accepting all reports. We therefore only have to consider the fusion center taking reports from all $K$ sensors for detection in what follows.

The $K$ sensors are divided into three groups. Each of the first two groups consists of $M$ sensors, while the last group contains $|\cN|-M$ sensors. All sensors in the first and third groups are honest, while those in the second group are compromised. Assume that there is a genie giving away the identities of $|\cN|-M$ honest sensors to the fusion center. For the rest $M$ honest sensors and $M$ compromised sensors, the identities are unknown to the fusion center. Without loss of generality, we assume that sensors in the first two groups have indices $[2M]$. We also give the observations used at each sensor (fake observations if the sensor is compromised) for generating its local report and the densities $P_{0,2}$ and $P_{1,2}$ to the fusion center. Let $s:[K]\rightarrow [2]$ be a function that assigns each sensor to index 1 or 2 (meaning ``honest'' or ``compromised'') in such a way that exactly $M$ out of the first $2M$ sensors are assigned to index 2, and the last $|\cN|-M$ sensors are all assigned to index 1. Let $\mathcal{S}$ be the collection of all possible assignments $s$. Clearly, there are total $|\mathcal{S}|=\binom{2M}{M}$ such assignments. For $\theta\in\{0,1\}$, the product density under the compromised group assignment $s$ is
\begin{equation}\label{eqn:prod_distribution}
    P_{\theta,s}(\mathbf{X}_t) = \prod_{k'=1}^{2M} P_{\theta,s(k')}(X_t^{k'}) \prod_{k=2M+1}^K P_{\theta,1}(X_t^{k}).
\end{equation}
Now, we are facing a composite change detection problem, which we refer to as genie-aided Byzantine centralized QCD (BCQCD). Before the change time $\nu$, the random vectors $\mathbf{X}_1, \mathbf{X}_2,
\ldots, \mathbf{X}_{\nu}$ are i.i.d. over time with density $P_{0,s}$
while $\mathbf{X}_{\nu+1},\mathbf{X}_{\nu+1}, \ldots$ are generated
with density $P_{1,s}$, for some $s \in \mathcal{S}$.
In this genie-aided version, the fusion center knows everything about the compromised sensors except for their exact locations. With slight abuse of notations, as \eqref{eq:delay_binary},
the mean detection delay of this problem is given by
\begin{equation} \label{eq_composite_delay}
    \cD_{\textrm{genie}}[T]:= \sup_{s\in\mathcal{S},\nu} \esssup\E^{s}_{\nu}[ (T-\nu)^+ | \mathbf{X}_1^{\nu} ];
\end{equation}
also as \eqref{eq:falseAlarmTime_multiple}, the mean time to false
alarm  is
\begin{equation} \label{eq_composite_alarm}
  \A_{\textrm{genie}}[T]:= \inf_{s \in \mathcal{S}} \E^{s}_{\infty}[ T].
\end{equation}

\begin{example}
    An example of genie-aided BCQCD with $|\cN|=4$ honest sensors and $M=2$ compromised sensors is provided in Fig.~\ref{fig:group_assignment}. In this figure, we use empty circles, crossed circles, and gray circles to represent honest sensors, compromised sensors, and honest sensors whose identities are revealed by the genie, respectively. The distribution that each sensor's observation follows under hypothesis $\theta\in\{0,1\}$ is presented. In both Fig.~\ref{fig:group_assignment}-a) and Fig.~\ref{fig:group_assignment}-b), we note that the last $|\cN|-M=2$ sensors are always honest and their identities are revealed to the fusion center. Hence, their observations always follow $P_{\theta}$ independently. For the first $2M=4$ sensors, the distributions under assignments $\{1,2,2,1\}$ and $\{1,1,2,2\}$ are shown in Fig.~\ref{fig:group_assignment}-a) and Fig.~\ref{fig:group_assignment}-b), respectively. We note that there are total $\binom{2M}{M}= 6$ different assignments and we only show 2 of them for demonstration.
\end{example}

\begin{figure}
    \centering
    \includegraphics[width=3in]{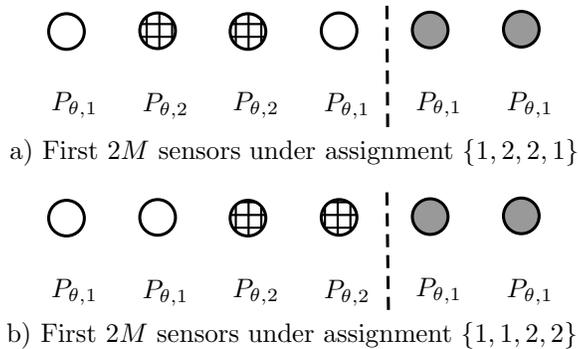}
    \caption{An example of genie-aided BCQCD with $|\cN|=4$ and $M=2$ under two different assignments. Here, empty circles and crossed circles are honest and compromised sensors under the assignment $s$, respectively, and thus follow $P_{\theta,1}$ and $P_{\theta,2}$, respectively. Moreover, gray circles are those honest sensors whose identities are revealed to the fusion center by the genie; thereby, their observations follow $P_{\theta,1}$ always, regardless of assignment.}
    \label{fig:group_assignment}
\end{figure}

\subsubsection{Transformed Centralized QCD}\label{sec:tQCD}
We transform the genie-aided BCQCD problem into an equivalent centralized QCD problem for which CUSUM is known to be optimal. Recall that unlike \cite{IHsiang18}, now only $2M$ sensors' identities are unknown to the fusion center and we define ``masked"-symmetric strategy as follows. Let $\tau_{2M}(\mathbf{X}_t)$ be the masked ordering map that puts the first $2M$ elements of its input $\mathbf{X}_t$ in descending order while keeps the other $|\cN|-M$ positions unchanged. A decision rule $T(.)$ of the genie-aided BCQCD problem is said to be masked symmetric if it can be represented as $T(\{\mathbf{X}_t\}_{t \geq 1})= \tilde{T}(\{\tau_{2M}(\mathbf{X}_t)\}_{t \geq 1})$ for some decision rule $\tilde{T}$.

In the transformed centralized QCD problem, the fusion center observes $\tilde{\mathbf{X}}_t=\tau_{2M}(\mathbf{X}_t)$ at time $t$. Let $\tilde{P}_{\theta}(\tilde{\mathbf{X}}_t)$ be the density of $\tau_{2M}(\mathbf{X}_t)$, where $\mathbf{X}_t$ is generated according to density $P_{\theta, s}, \theta\in\{0,1\}$.  Before the change, the observations $\{\tilde{\mathbf{X}}_t\}$ follow $\tilde{P}_{0}$ while after the change, they follow $\tilde{P}_{1}$. Also,
\begin{align}\label{eq_aux_PDF}
    \tilde{P}_{\theta}(\tilde{\mathbf{X}}_t) = \sum_{\mathbf{X}_t: \tau_{2M}(\mathbf{X}_t)=\tilde{\mathbf{X}}_t }P_{\theta, s}\left( \mathbf{X}_t \right),
\end{align}
where the equality follows from that the absolute value of the
Jacobian of a permutation is always 1. Following the proof of part
1 of \cite[Lemma 4.1]{IHsiang18}, we can easily show that for all
assignments $s \in \mathcal{S}$, the density
$\tilde{P}_{\theta}(\tilde{\mathbf{X}}_t)$ does not depend on $s$.
Suppose the change occurs at the time $\nu$. Under hypothesis
$\tilde{H}_1$, the random vectors $\tilde{\mathbf{X}}_1,
\tilde{\mathbf{X}}_2, \ldots, \tilde{\mathbf{X}}_{\nu}$ are drawn
i.i.d. over time with density $\tilde{P}_{0}$ while
$\tilde{\mathbf{X}}_{\nu+1},\tilde{\mathbf{X}}_{\nu+2}, \ldots$
are generated i.i.d. with density $\tilde{P}_{1}$. Under
hypothesis $\tilde{H}_0$, there is no change, i.e. $\nu=\infty$,
and $\tilde{\mathbf{X}}_t$ are drawn i.i.d. with density
$\tilde{P}_{0}$ for all $t$.

We first focus on a masked symmetric rule $T$, and show that the detection delay of the genie-aided BCQCD is identical to that of the transformed QCD problem, defined as $\cD_{\textrm{trans}}[\tilde{T}]:=\sup_\nu\esssup\E_{\nu}[(\tilde{T}-\nu)^+|\tilde{\mathbf{X}}_1^{\nu}]$. Compared to the transformation in the one-shot hypothesis testing \cite{IHsiang18}, our delay in \eqref{eq_composite_delay} involves all pre-change observations by taking an essential supreme over the distributions of them. This difference complicates the transformation. Specifically, we will show
\begin{align} \label{eq_no_infProd}
  \cD_{\textrm{genie}}[T] &= \sup_{s \in \mathcal{S},\nu} \esssup\E^{s}_{\nu}[ (\tilde{T}(\{\tau_{2M}(\mathbf{X}_t)\}_{t \geq 1}) -\nu)^+ |\mathbf{X}_1^{\nu}]  \notag \\
  &\overset{(a)}{=}\sup_{\nu}\esssup\E_{\nu}[ (\tilde{T}(\{\tilde{\mathbf{X}}_t\}_{t \geq 1})-\nu)^+ |\tilde{\mathbf{X}}_1^{\nu}]=\cD_{\textrm{trans}}[\tilde{T}].
\end{align}
The first equality is from the definition in \eqref{eq_composite_delay} and we will devote ourselves to proving equality (\ref{eq_no_infProd}a). To this end, note that
\begin{align}\label{eq_no_infProd1}
&\E^s_{\nu}[ (\tilde{T}(\{\tau_{2M}(\mathbf{X}_t)\}_{t\geq 1})-\nu)^+|\mathbf{X}_1^ \nu  ]    \notag\\
=  &\sum_{z=0}^\infty 1-\P\left(\left(\tilde{T}(\{\tau_{2M}(\mathbf{X}_t)\}_{t \geq 1})-\nu \right)^+ \leq z \big|\mathbf{X}_1^ \nu \right) \notag \\
=  &\sum_{z=0}^\infty 1- \! \int \! 1_{\left\{ \left(\tilde{T}(\{\tau_{2M}(\mathbf{x}_t)\}_{t \geq 1})-\nu \right)^+ \leq z \right\}} \! \prod_{t=\nu+1}^{\nu+z} \!\!\! P_{1,s}(\mathbf{x}_t)d{\mathbf{x}}_{\nu+1}^{\nu+z} \notag \\
\overset{(a)}{=}  &\sum_{z=0}^\infty 1- \int 1_{\left\{\left(\tilde{T}(\{\tilde{\mathbf{x}}_t\}_{t \geq 1})-\nu \right)^+ \leq z \right\}} \prod^{\nu+z}_{t = \nu+1} \tilde{P}_{1}(\tilde{\mathbf{x}}_t)d\tilde{\mathbf{x}}_{\nu+1}^{\nu+z}\notag \\
=  &\sum_{z=0}^\infty 1- \P\left(\left(\tilde{T}(\{\tilde{\mathbf{X}}_t\}_{t \geq 1})-\nu \right)^+ \leq z \big|\tilde{\mathbf{X}}_1^ \nu \right) \notag \\
=  &\E_{\nu}[ (\tilde{T}(\{\tilde{\mathbf{X}}_t\}_{t \geq 1})-\nu)^+ |\tilde{\mathbf{X}}_1^{\nu}],
\end{align}
where $\P$ is the associated probability measure and $1_{\{.\}}$
is the indicator function; and (a) follows from the change of
variables in integration \cite{billingsley2008probability} and the
fact that $\tilde{P}_{1}(.)$ does not depend on $s$.

Now, for a fixed $\nu$ and for each $s \in\mathcal{S}$, let $\P^s$ and $\tilde{\P}$ denote the probability measures on $\R^{K\times\nu}$ with densities specified by \eqref{eqn:prod_distribution} and \eqref{eq_aux_PDF} with $\theta=0$, respectively.
To establish (\ref{eq_no_infProd}a), observe that for any $x\in
\R^{K\times\nu}$, from \eqref{eq_no_infProd1}, we have
\begin{align}\label{1}
& \E^{s}_{\nu}[ (\tilde{T}(\{\tau_{2M}(\mathbf{X}_t)\}_{t \geq 1}) -\nu)^+ |\mathbf{X}_1^{\nu}] (x) =\notag \\
& \E_{\nu}[ (\tilde{T}(\{\tilde{\mathbf{X}}_t\}_{t \geq 1})-\nu)^+ |\tilde{\mathbf{X}}_1^{\nu}] (\tau_{2M}(x)).
\end{align}
Let $D_M$ denote $\esssup\E^{s}_{\nu}[
(\tilde{T}(\{\tau_{2M}(\mathbf{X}_t)\}_{t \geq 1}) -\nu)^+
|\mathbf{X}_1^{\nu}]$ where the essential supremum is taken under
$\P^s$. By definition, there exists
$\Omega\subseteq\R^{K\times\nu}$ with $\P^s (\Omega)=1$ such that
\[
D_M \ge \E^{s}_{\nu}[ (\tilde{T}(\{\tau_{2M}(\mathbf{X}_t)\}_{t
\geq 1}) -\nu)^+ |\mathbf{X}_1^{\nu}](x)
\]
for all $x\in \Omega$. By \eqref{1}, we have
\[
D_M \ge \E_{\nu}[ (\tilde{T}(\{\tilde{\mathbf{X}}_t\}_{t \geq 1})-\nu)^+ |\tilde{\mathbf{X}}_1^{\nu}] (\tau_{2M}(x)),\quad \forall x\in \Omega.
\]
Note that
\[
\tilde{\P}(\tau_{2M}(\Omega)) = \int_{\tau_{2M}(\Omega)} \tilde{P}_{0}(y) dy =  \int_{\Omega} P_{0,s}(x) dx =  \P^{s} (\Omega)=1,
\]
where the densities $\tilde{P}_{0}$ and $P_{0,s}$ are given by \eqref{eq_aux_PDF} and \eqref{eqn:prod_distribution} respectively. We therefore conclude that $D_M \ge \esssup\E_{\nu}[ (\tilde{T}(\{\tilde{\mathbf{X}}_t\}_{t \geq 1})-\nu)^+ |\tilde{\mathbf{X}}_1^{\nu}]$ where the essential supremum here is taken under $\tilde{\P}$. Noting that this is true for every $s\in\mathcal{S}$ gives the relation ``$\geq$'' in (\ref{eq_no_infProd}a). By using the same argument as above, but switching the roles of the left-hand side and right-hand side of (\ref{eq_no_infProd}a), we obtain the relation ``$\le$''.

We have shown that the detection delay of genie-aided BCQCD \eqref{eq_composite_delay} is equal to that of transformed QCD under masked symmetric rules. One can similarly prove that the mean time to false alarm $\mathcal{A}_{\textrm{genie}}[T]$ in \eqref{eq_composite_alarm} is equal to that of the new problem $\mathcal{A}_{\textrm{trans}}[\tilde{T}]:=\mathbb{E}_{\infty}[\tilde{T}]$. The rest is to show that for any fusion rule $T'(.)$, there is a masked symmetric rule $T(.)$ that is not worse than $T'(.)$. This is shown in Lemma~\ref{lem:general} in Appendix and then the transformation of QCD is established.

\subsubsection{Establishing the converse}\label{sec:est_converse}
Let $T^*$, $T_{\textrm{genie}}^*$, and $T_{\textrm{trans}}^*$ be
optimal stopping rules for BDQCD, genie-aided BCQCD, and
transformed QCD, respectively. So far, we have established the
following relationships among the aforementioned problems
\begin{equation} \label{eq_bin_bound_Trans}
    \mathcal{D}[T^*]\geq \mathcal{D}_{\textrm{genie}}[T_{\textrm{genie}}^*] = \mathcal{D}_{\textrm{trans}}[T_{\textrm{trans}}^*],
\end{equation}
and
\begin{equation} \label{eq_bin_bound_Trans_alarm}
    \mathcal{A}[T^*]\leq \mathcal{A}_{\textrm{genie}}[T_{\textrm{genie}}^*] = \mathcal{A}_{\textrm{trans}}[T_{\textrm{trans}}^*].
\end{equation}

We can therefore establish a converse bound by evaluating the performance of the transformed QCD, which is a standard QCD problem with observations following $\tilde{P}_0$ and $\tilde{P}_1$ before and after the change point $\nu$, respectively. For such the problem, it is well known from \cite{moustakides86}, \cite[Lemma 2]{bounds} that an optimal strategy is Page's CUSUM procedure given by $\tilde{\sigma}(h)=\inf\{t\in\mathbb{N}:\tilde{Y}_t\geq h\}$, where $\tilde{Y}_t = (\tilde{Y}_{t-1}+\tilde{\ell}_t)^+$ with $\tilde{Y}_0=0$, and from \eqref{eq_aux_PDF}
\begin{align}\label{eqn:cusum_to_mcusum}
  \tilde{\ell}_t &= \log\frac{ \sum_{\mathbf{X}_t:\tau_{2M}(\mathbf{X}_t)=\tilde{\mathbf{X}}_t} P_{1,s}(\mathbf{X}_t) }{ \sum_{\mathbf{X}_t:\tau_{2M}(\mathbf{X}_t)=\tilde{\mathbf{X}}_t} P_{0,s}(\mathbf{X}_t) }\nonumber \\
   &= \log\frac{ \sum_{\pi\in\Pi_{2M}} P_{1,s}(\pi(\tilde{\mathbf{X}}_t)) }{ \sum_{\pi\in\Pi_{2M}} P_{0,s}(\pi(\tilde{\mathbf{X}}_t)) } \textcolor{black}{= \log\frac{ \sum_{\pi\in\Pi_{2M}} P_{1,s\circ\pi^{-1}}(\tilde{\mathbf{X}}_t) }{ \sum_{\pi\in\Pi_{2M}} P_{0,s\circ\pi^{-1}}(\tilde{\mathbf{X}}_t) } } \nonumber \\
   &\overset{(a)}{=} \log\frac{ \sum_{s'\in\mathcal{S}} P_{1,s'}(\tilde{\mathbf{X}}_t) \frac{ 2M!}{\binom{2M}{M}} }{ \sum_{s'\in\mathcal{S}} P_{0,s'}(\tilde{\mathbf{X}}_t)\frac{2M!}{\binom{2M}{M}} }.
\end{align}
where $\pi:[K]\rightarrow[K]$ is a masked permutation function that permutes the first $2M$ entries while keeps the remaining $|\cN|-M$ entries unchanged, $\Pi_{2M}$ is the collection of all ($2M!$ in total) such $\pi$, and $\circ$ is the function composition operator; (a) follows from the fact that for a compromised group assignment $s$, summing over all the permuted versions $s \circ\pi^{-1}$ is equivalent to summing over all the assignments $s'$ with each $s'$ being involved $2M!/\binom{2M}{M}$ times. In what follows, we set $P_{0,2}=P_{1,1}$ and $P_{1,2}=P_{0,1}$ according to the reverse attack described in Sec. \ref{sec:rev_attack}. We now rewrite the likelihood in \eqref{eqn:cusum_to_mcusum} as
\begin{align}
  \tilde{\ell}_t &= \log\frac{\sum_{s\in\mathcal{S}} P_{1,s}(\tilde{\mathbf{X}}_t)}{\sum_{s\in\mathcal{S}} P_{0,s}(\tilde{\mathbf{X}}_t)} \label{eq_mixtureCUSUM} \\
    &\overset{(a)}{=} \log\frac{\left(\sum_{s\in\mathcal{S}} \prod_{k'=1}^{2M} P_{1,s(k')}(\tilde{X}_t^{k'})\right) \prod_{k=2M+1}^K P_{1,1}(\tilde{X}_t^{k})}{\left(\sum_{s\in\mathcal{S}} \prod_{k'=1}^{2M} P_{0,s(k')}(\tilde{X}_t^{k'})\right) \prod_{k=2M+1}^K P_{0,1}(\tilde{X}_t^{k})} \nonumber \\
    &\overset{(b)}{=} \log\frac{\prod_{k=2M+1}^K P_{1,1}(\tilde{X}_t^{k})}{\prod_{k=2M+1}^K P_{0,1}(\tilde{X}_t^{k})},\label{eq_mixtureCUSUM1}
\end{align}
where (a) follows from \eqref{eqn:prod_distribution} and (b) is
because of the fact that for every $s$, there exists a $\bar{s}$
such that $\bar{s}(k')=2$ whenever $s(k')=1$ and $\bar{s}(k')=1$
whenever $s(k')=2$; therefore,
\begin{align}
    \sum_{s\in\mathcal{S}}\prod_{k'=1}^{2M} P_{1,s(k')}(\tilde{X}_t^{k'}) &= \sum_{s\in\mathcal{S}} \prod_{k'=1}^{2M} P_{0,\bar{s}(k')}(\tilde{X}_t^{k'})\nonumber \\
    &= \sum_{\bar{s}\in\mathcal{S}} \prod_{k'=1}^{2M} P_{0,\bar{s}(k')}(\tilde{X}_t^{k'}),
\end{align}
where the first equality is from $P_{1,1}=P_{0,2}$ and $P_{1,2}=P_{0,1}$. Note that when $k=2M+1 \ldots K$, $\tilde{X}_t^{k}$ is equal to the honest observation $X_t^{k}$ in the BCQCD before transformation. Hence, the optimal test reduces to the standard centralized CUSUM procedure for the change detection with $|\cN|-M$ honest sensors. Applying the results in \cite{bounds} then shows \eqref{eq:new_converse}. 

\subsubsection{Discussions}\label{sec:Conclude}
A byproduct obtained along the proof is that the ``reverse attack" proposed in Section~\ref{sec:rev_attack} is an asymptotically worst attack for the original BDQCD problem. This observation is further exploited using the game theory, as described in Lemma \ref{lma:scaling_lower_bound} and Theorem \ref{Thm_Binarygame_cost} of Section \ref{sec:game}.

Note that although the idea of the transformation in Sec.
\ref{sec:tQCD} is inspired by \cite{IHsiang18}, our proofs
presented above are quite different. In \cite{IHsiang18}, all
sensors' identities are not revealed to the fusion center, while
it is essential for us to construct a strong, but not so powerful,
genie that reveals identities of some sensors for proving a tight
converse, as in Sec. \ref{sec:genie}. This new genie is the key to
validate \eqref{eq_mixtureCUSUM1} in Sec. \ref{sec:est_converse},
which shows that the optimal test under the proposed reverse
attack constructed in Sec. \ref{sec:rev_attack} only relates to
the revealed $|\cN|-M$ honest sensors. Furthermore, the
transformation in Sec. \ref{sec:tQCD} is more involved than that
in \cite{IHsiang18}. The difficulty comes from the fundamental
difference between the one-shot hypothesis testing problem in
\cite{IHsiang18} and our sequential change detection that deals
with observation sequences with an unknown change time. Finally,
in the coming Section \ref{subsec:converse}, we extend our
converse to the multiple-hypothesis case and face new challenges
compared with the binary hypothesis problem in \cite{IHsiang18}.

\section{Multi-Hypothesis BDQCD}\label{sec:mul_BDQCD}
In this section, we consider the multi-hypothesis BDQCD. Again, we first present our main result of this section in Section~\ref{subsec:result_M}, which is the characterization of the asymptotic performance of the worst-case detection delay subject to a mean time to a false alarm or a false isolation. We then prove a converse for the considered problem in Section~\ref{subsec:converse}, followed by the proposed stopping rules and their performance analysis in Section~\ref{subsec:stop_rule_M}. Throughout the section, we define for each pair of $q,j\in[Q]^+$, $q\neq j$, the KL divergence from $P_j$ to $P_q$ as
\begin{equation} \label{eq_KL_eachsensor}
 I(q,j) := \int \log\left(P_q(x)/P_j(x)\right) P_q(x)dx.
\end{equation}
 Let $\sigma^2(q,j)$ be the second moment of $I(q,j)$ defined as
\begin{equation}
    \sigma^2(q,j) := \E_q\left[ \left(\log\left(\frac{P_q(x)}{P_j(x)}\right) - I(q,j) \right)^2\right].
\end{equation}
We then make the following assumption:
\begin{assumption}\label{asm:I}
For any $q\in[Q]$,
\begin{itemize}
\item [(i)] $0<I(q,j) <\infty$ and $\sigma^2(q,j) <\infty$, $\forall j\in [Q]^+, j\neq q$.
\item [(ii)] Let $I^q := \min_{0\le j\le Q,\ j\neq q} I(q,j)$. Assume $I^q$ admits a unique minimizer $j_q^*\in [Q]^+ \setminus\{q\}$.
\end{itemize}
Now, consider $q=0$.
\begin{itemize}
    \item [(iii)] We define $I^0 := \min_{1\le j \le Q} I(j,0)$, and assume that $I^0$ admits a unique minimizer $j_0^*\in [Q]$.
\end{itemize}
\end{assumption}

\subsection{Main results of this section}\label{subsec:result_M}
Our first result for multi-hypothesis BDQCD with $Q+1$ hypotheses is the characterization of the converse as follows.
\begin{theorem}\label{thm:scaling_M}
   Consider the multi-hypothesis BDQCD with $Q+1$ hypotheses. For any rule $T$ with $\A[T] \geq \gamma$, the worst-case mean detection delay is lower bounded as
    \begin{equation} \label{eqn:scaling_M}
        \mathcal{D}[T]\gtrsim \frac{\log\gamma}{(|\cN|-M)I^*}, \quad\text{as $\gamma\rightarrow\infty$}.
    \end{equation}
    where
    \begin{equation} \label{eqn:I*}
    I^*=\min_q I^q= \min_{q \in [Q]} \min_{j \in [Q]^+ \setminus\{q\}} I(q,j)
    \end{equation}
\end{theorem}
\noindent This converse can be proved in a similar way to Theorem~\ref{thm:new_converse}; hence, we only list the major differences between the two proofs in Section~\ref{subsec:converse}.

For the achievability, we propose in Section~\ref{subsec:stop_rule_M} a family of stopping rules, called the simultaneous $d$-th alarm $\tau_{(d)}^s$, and show that this stopping rule achieves the first-order scaling of \eqref{eqn:scaling_M} when $d$ is set to be $|\cN|$. This simultaneous $d$-th alarm rule requires each sensor to send a $Q$-bit signal constantly through the noiseless link to the fusion center.
To reduce the demanding bandwidth and energy requirements, another family of stopping rules, called multi-shot $d$-th alarm $\tau_{(d)}^m$, is proposed in Section~\ref{subsec:stop_rule_M}. This rule is more economic in that it only requires the sensor sending $\lceil\log_2Q\rceil$-bit signal occasionally. In what follows, we present the asymptotic performance of the two proposed families of rules and refer the reader to Section~\ref{subsec:stop_rule_M} for their proofs.
\begin{theorem}\label{thm:scaling_M_achiv}
    With the wort-case mean time to a false alarm or isolation no smaller than than $\gamma$, we have\\
    (a) among the proposed simultaneous $d$-th alarm rule $\tau_{(d)}^s$, the best first-order asymptotic worst-case mean detection delay is achieved when $d=|\cN|$ and is given by
    \begin{equation}
        \mathcal{D}[\tau_{(|\cN|)}^s] \lesssim \frac{\log\gamma}{(|\cN|-M)I^*}, \quad\text{as $\gamma\rightarrow\infty$};
    \end{equation}
    (b) among the proposed multi-shot $d$-th alarm rule $\tau_{(d)}^m$, the first-order asymptotic worst-case mean detection delay when $d \geq M+1$ is given by
    \begin{equation}
        \mathcal{D}[\tau_{(d)}^m] \lesssim \frac{\log\gamma}{I^*}, \quad\text{as $\gamma\rightarrow\infty$}.
    \end{equation}
\end{theorem}

Combining the results in Theorems \ref{thm:scaling_M} and \ref{thm:scaling_M_achiv}, we arrive at the following result of the optimal scaling for multi-hypothesis BDQCD with $Q+1$ hypotheses.
\begin{corollary}\label{thm:opt_scaling_M_achiv}
For multi-hypothesis BDQCD with $Q+1$ hypotheses and the number of honest sensors $|\cN|\geq M+1$, if the noiseless link of each sensor can support (at least) $Q$ bits, the first-order asymptotic worst-case mean detection delay of an optimal stopping rule $T^*$ subject to $\A[T^*] \geq \gamma$ is precisely
 \begin{equation}
        \mathcal{D}[T^*] \sim \frac{\log\gamma}{(|\cN|-M)I^*}, \quad\text{as $\gamma\rightarrow\infty$}.
    \end{equation}
Moreover, if $|\cN|=M+1$, the optimal scaling
 \begin{equation}
        \mathcal{D}[T^*] \sim \frac{\log\gamma}{I^*}, \quad\text{as $\gamma\rightarrow\infty$},
    \end{equation}
    can be achieved by a stopping rule that requires only  $\lceil\log_2Q\rceil$-bit links.
 \end{corollary}

\subsection{Proof of Theorem \ref{thm:scaling_M}, the converse for multi-hypothesis BDQCD in Theorem \ref{thm:scaling_M_achiv}}\label{subsec:converse}

For extending the converse from the binary case to the
multi-hypothesis case with $Q+1$ hypotheses, we encounter two main
challenges. First, the asymptotically worst attack adopted in the
binary case, namely the reverse attack, cannot be
straightforwardly applied. As there are $Q$ post-change
distributions, we have to carefully choose one of them for
swapping in order to make the attack strategy asymptotically
worst. Second, after we manage to construct the attack strategy
and complete the transformation, there are $Q+1$ hypotheses in the
transformed QCD and hence $Q^2$ LLRs (see
\eqref{eq_CUSUM_element}) to be tracked in the asymptotically
optimal matrix CUSUM procedure. It is difficult to make each LLR
relate to only observations of (a subset of) honest sensors as we
have done in \eqref{eq_mixtureCUSUM1} for the binary case.

In what follows, to solve the first issue, we modify the reverse attack in upcoming Section \ref{sec:mul_rev_attack} by swapping the $P_0$ with the post-change distribution that is closest to $P_0$ in the sense of having the minimum KL divergence. To circumvent the second issue, we abandon the approach of evaluating the optimal detecting procedure in Section~\ref{sec:est_converse} and directly perform the delay analysis based on \cite[Theorem 2]{nikiforov95}.

\subsubsection{The reverse attack for $(Q+1)$-hypotheses}\label{sec:mul_rev_attack}
To prove this converse, under the true hypothesis $H_q, q\in[Q]^+$ (defined in Section \ref{subsec:M_BDQCD}), we define $P_{q,1}=P_{q}$. The distribution $P_{q,2}$ for fake observation in the proposed attack strategy is constructed as follows. Let $q_m=\arg_{q \in [Q]^+} \min I_q$, with the corresponding $j^*_{q_m}$ defined in Assumption \ref{asm:I} (ii) and (iii). We define
\begin{equation}
    P_{q,2} = \begin{cases}
    P_{j_q^*}, \quad &\hbox{if $q\neq j^*_{q_m}$,} \\
    P_{q_m},\quad &\hbox{if $q= j^*_{q_m}$.} \label{eq_mul_rev_attacksym}
    \end{cases}
\end{equation}
The main intuition behind this choice is that we want the compromised sensors to follow the distribution that is closest to the true one under the KL divergence. This is done in the first case above. For the exception in the second case, i.e., $q= j^*_{q_m}$, it is for having the symmetry $P_{j^*_{q_m},2}=P_{q_m}$ and $P_{q_m,2}=P_{j_{q_m}^*}$, which will become handy later in the proof.


\subsubsection{Genie-aided Byzantine centralized $(Q+1)$-hypotheses QCD}\label{sec:mul_genie}
\textcolor{black}{To establish the tight converse, we assume that a genie gives the fusion center the identities of $|\cN|-M$ out of $|\cN|$ honest sensors. Without loss of generality, we let these $|\cN|-M$ sensors have the indices $2M+1,\ldots,K$, respectively. The genie also provides the fusion center with the observations used at each sensor for generating its local reports. Similarly to \eqref{eqn:prod_distribution} in the binary case, with the help of this genie, the problem becomes the genie-aided BCQCD with distribution for each $q\in[Q]^+$ as
\begin{equation}\label{eq:til_P_s}
    P_{q,s}(\mathbf{X}_t)=\prod_{k'=1}^{2M} P_{q,s(k')}(X_t^{k'})\prod_{2M+1}^{K} P_{q,1}(X_t^{k}).
\end{equation}
For this genie-aided BCQCD, replacing $g\in\mathcal{G}$ with $s\in\mathcal{S}$ in \eqref{eq:ess_delay_multiple} and \eqref{eq:falseIsolationTime_multiple}, we obtain the detection delay and mean time to a false alarm or a false isolation given by
\begin{equation} \label{eq_composite_delay_multi}
    \cD_{\textrm{genie}}[T] = \sup_{q \in [Q]} \sup_{s\in\mathcal{S},\nu}\esssup \mathbb{E}^{q,s}_{\nu}[ (T-\nu)^+|\mathbf{X}_1^{\nu}],
\end{equation}
and
\begin{equation} \label{eq_composite_false_multi}
    \A_{\textrm{genie}}[T] = \inf_{q \in [Q]^+} \inf_{s\in\mathcal{S}} \inf_{\hat{q}\in[Q]\setminus\{q\}}\mathbb{E}_0^{q,s}[ T^{\hat{q}} ],
\end{equation}
respectively.
}
\subsubsection{Transformed Centralized $(Q+1)$-hypotheses QCD}\label{sec:mul_tQCD}
We can now follow the proof of the binary case to transform the genie-aided BCQCD into a multiple-hypothesis QCD having distributions generalizing \eqref{eq_aux_PDF} as
\begin{equation}\label{eq:til_P}
    \tilde{P}_{q}(\tilde{\mathbf{X}}_t) = \sum_{\mathbf{X}_t:\tau_{2M}(\mathbf{X}_t)=\tilde{\mathbf{X}}_t} P_{q,s}(\mathbf{X}_t),
\end{equation}
where $q\in[Q]^+$ and $\tau_{2M}$ is again the masked ordering map. Let us define the KL divergence between transformed distributions $\tilde{P}_{q}$ and $\tilde{P}_{j}$ for $q\in[Q]$ and $j\in[Q]^+$ as
\begin{equation}\label{eq:til_I}
    \tilde{I}(q,j) = \int\log\left(\frac{\tilde{P}_{q}(\tilde{\mathbf{x}})}{\tilde{P}_{j}(\tilde{\mathbf{x}})}\right)\tilde{P}_{q}(\tilde{\mathbf{x}})d\tilde{\mathbf{x}}.
\end{equation}
Moreover, let $\tilde{I}_q^*=\min_{j \in [Q]^+ \setminus\{q\}} \tilde{I}(q,j)$ and
\begin{equation}\label{eq:til_I*}
    \tilde{I}^* = \min_{q\in[Q]} \tilde{I}_q^* = \min_{q\in[Q]}\min_{j \in [Q]^+ \setminus\{q\}} \tilde{I}(q,j).
\end{equation}

Note that in the binary case \eqref{eq_bin_bound_Trans}, we have shown that the transformed QCD would have the same detection delay with the genie-aided BCQCD for any post-change distribution $P_1$. Extending from the binary case \eqref{eq_composite_delay} to the multi-hypothesis case \eqref{eq_composite_delay_multi}, we just need to take an additional supremum over all post-change distributions $P_{q}$, $q\in[Q]$; therefore, the equivalence
\begin{align}
    \cD_{\textrm{genie}}[T_{\textrm{genie}}^*] &= \cD_{\textrm{trans}}[T_\textrm{trans}^*]  \nonumber \\
    &:=\sup_{q \in [Q]} \sup_{\nu}\esssup \mathbb{E}^{q}_{\nu}\left[ \left( T_\textrm{trans}^*(\{\tilde{\mathbf{X}}_t\}_{t\geq 1})-\nu \right)^+ \bigg| \tilde{\mathbf{X}}_1^{\nu} \right],
\end{align}
still holds under multiple hypotheses for optimal rules $T_\textrm{genie}^*$ and $T_\textrm{trans}^*$  in genine-aided BCQCD and transformed QCD, respectively.

Regarding the mean time to a false alarm or false isolation, it is a bit more involved than the proof for detection delay since in addition to false alarm considered in Section \ref{sec:tQCD}, we also need to deal with false isolation. {We again first focus on a masked symmetric rule $T$ satisfying $T(\{\mathbf{X}_t\}_{t\geq 1}) = \tilde{T}(\{\tau_{2M}(\mathbf{X}_t)\}_{t\geq 1})$ for some rule $\tilde{T}$. We aim to prove
\begin{equation} \label{eq_multi_trans_false}
\A_{\textrm{genie}}[T]=\A_{\textrm{trans}}[\tilde{T}]
\end{equation}
Consider $\A_{\textrm{genie}}[T]$ in \eqref{eq_composite_false_multi}, for each $q\in[Q]^+$ and $s\in\mathcal{S}$, we have
\begin{align}
  &\inf_{\hat{q}\in[Q]\setminus\{q\}}\mathbb{E}^{q,s}_{0}[ T^{\hat{q}} ]= \inf_{\hat{q}\in[Q]\setminus\{q\}}\mathbb{E}^{q,s}_{0} \left[ \inf_{\rho\geq 1} \frac{T_{\rho}}{1_{\{\hat{q}_{\rho}=\hat{q}\}}} \right] \nonumber \\
  &=\inf_{\hat{q}\in[Q]\setminus\{q\}}\sum_{z=0}^{\infty}1-\mathbb{P}\left( \inf_{\rho\geq 1} \frac{T_{\rho}}{1_{\{\hat{q}_{\rho}=\hat{q}\}}} \leq z \right), \label{eq:symmetric_false_iso}
  \end{align}
where the first equality is from \eqref{eq:stopping_rule_multi} and the convention $a/0:=\infty$ for positive $a\in\mathbb{R}$. Observe that
\begin{align}\label{eq:P_Ez}
    \mathbb{P}\left( \inf_{\rho\geq 1} \frac{T_{\rho}}{1_{\{\hat{q}_{\rho}=\hat{q}\}}}\leq z \right)&=\int 1_{\left\{\inf_{\rho\geq 1} \frac{T_{\rho}(\{\mathbf{X}_t\}_{t\geq 1})}{1_{\{\hat{q}_{\rho}=\hat{q}\}}}\leq z\right\}} \prod_{t=1}^z P_{q,s}(\mathbf{x}_t) d\mathbf{x}_1^z \nonumber \\
    &= \int 1_{\left\{\inf_{\rho\geq 1} \frac{\tilde{T}_{\rho}(\{\tau_{2M}(\mathbf{X}_t)\}_{t\geq 1})}{ 1_{\{\hat{q}_{\rho}=\hat{q}\}}}\leq z\right \}}  \prod_{t=1}^z P_{q,s}(\mathbf{x}_t) d\mathbf{x}_1^z.
\end{align}
Plugging \eqref{eq:P_Ez} into \eqref{eq:symmetric_false_iso}, then \eqref{eq:symmetric_false_iso} equals to
\begin{align}
     &\inf_{\hat{q}\in[Q]\setminus\{q\}}\sum_{z=0}^{\infty}1-\int 1_{\left\{\inf_{\rho\geq 1} \frac{\tilde{T}_{\rho}(\{\tilde{\mathbf{X}}_t\}_{t\geq 1})}{ 1_{\{\hat{q}_{\rho}=\hat{q}\}}}\leq z\right \}} \prod_{t=1}^z \tilde{P}_{q,s}(\tilde{\mathbf{x}}_t) d\tilde{\mathbf{x}}_1^z \nonumber \\
    &= \inf_{\hat{q}\in[Q]\setminus\{q\}}\sum_{z=0}^{\infty}1- \mathbb{P}\left( \inf_{\rho\geq 1} \frac{\tilde{T}_{\rho}}{1_{\{\hat{q}_{\rho}=\hat{q}\}}}\leq z \right) \nonumber \\
    &= \inf_{\hat{q}\in[Q]\setminus\{q\}}\mathbb{E}^{q,s}_{0}\left[ \inf_{\rho\geq 1} \frac{\tilde{T}_{\rho}}{1_{\{\hat{q}_{\rho}=\hat{q}\}}}\right ] \nonumber \\
    &= \inf_{\hat{q}\in[Q]\setminus\{q\}}\mathbb{E}^{q,s}_{0}[ \tilde{T}^{\hat{q}} ].
\end{align}
Taking infimum over $q\in[Q]^+$ and $s\in\mathcal{S}$ on the above shows \eqref{eq_multi_trans_false} for a masked symmetric rule $T$. In Lemma \ref{lem:general_mult_false}, we show that it suffices to consider masked symmetric rules as for any general rule $T'$, there exists a symmetrized rule which is not worse than it in $\A_{\textrm{genie}}[T']$.

We have extended \eqref{eq_bin_bound_Trans_alarm} and again shown that
$\A_{\textrm{genie}}[T_{\textrm{genie}}^*]=\A_{\textrm{genie}}[T_{\textrm{trans}}^*]$
for optimal rules $T_{\textrm{genie}}^*$ and
$T_{\textrm{trans}}^*$ in the multi-hypothesis cases and completed the transformation. To
establish the converse, we now provide a converse to the
asymptotic performance of the transformed QCD in Lemma \ref{lem:extend_Nikogonov}. That is, the first-order scaling of an optimal stopping rule $T^*_{\mathrm{QCD}}$ with $\A_{\textrm{trans}}[T^*_{\mathrm{QCD}}] \geq \gamma$ is given by
\begin{equation} \label{Delay_tildeI*}
    \mathcal{D}_{\textrm{trans}}[T^*_{\mathrm{QCD}}]\sim \frac{\log\gamma}{\tilde{I}^*}.
\end{equation}

\subsubsection{Establishing the converse by evaluating the transformed delay} \label{sec:mul_est_converse}

To establish the converse, we now look into the structure of $\tilde{I}^*$ in  \eqref{Delay_tildeI*}, which is defined in \eqref{eq:til_I*}. First, we note from \eqref{eq:til_P}-\eqref{eq:til_I} that for any pair of hypothesis indexes $(q,j)$
\begin{align}\label{eq:til_I_qj}
    &\tilde{I}(q,j) 
    = \int \log\left(\frac{\sum_{\mathbf{x}:\tau_{2M}(\mathbf{x})=\tilde{\mathbf{x}}} P_{q,s}(\mathbf{x})}{\sum_{\mathbf{x}:\tau_{2M}(\mathbf{x})=\tilde{\mathbf{x}}} P_{j,s}(\mathbf{x})}\right)\tilde{P}_{q}(\tilde{\mathbf{x}}) d\tilde{\mathbf{x}} \nonumber \\
    &\overset{(a)}{=} \int \log\left(\frac{\sum_{s\in\mathcal{S}} P_{q,s}(\tilde{\mathbf{x}})\frac{2M!}{\binom{2M}{M}}}{\sum_{s\in\mathcal{S}} P_{j,s}(\tilde{\mathbf{x}})\frac{2M!}{\binom{2M}{M}}}\right) \tilde{P}_{q}(\tilde{\mathbf{x}}) d\tilde{\mathbf{x}} \nonumber \\
    & \overset{(b)}{=} \int \log\left(\frac{\sum_{s\in\mathcal{S}} \prod_{k'=1}^{2M} P_{q,s(k')}(\tilde{x}^{k'})\frac{2M!}{\binom{2M}{M}} }{ \sum_{s\in\mathcal{S}} \prod_{k'=1}^{2M} P_{j,s(k')}(\tilde{x}^{k'})\frac{2M!}{\binom{2M}{M}}}\right) \tilde{P}_{q}(\tilde{\mathbf{x}}) d\tilde{\mathbf{x}}  \nonumber \\
    & \hspace{0.2in}+\int \log\left(\frac{\prod_{2M+1}^{K} P_{q,1}(\tilde{x}^k)  }{\prod_{2M+1}^{K} P_{j,1}(\tilde{x}^k)}\right) \tilde{P}_{q}(\tilde{\mathbf{x}}) d\tilde{\mathbf{x}}  \nonumber \\
    & \overset{(c)}{=} \int \log\left(\frac{\hat{P}_{q,s}(\tilde{x}^1,\ldots,\tilde{x}^{2M})}{\hat{P}_{j,s}(\tilde{x}^1,\ldots,\tilde{x}^{2M})}\right) \hat{P}_{q,s}(\tilde{x}^1,\ldots,\tilde{x}^{2M}) d\tilde{x}^1\ldots d\tilde{x}^{2M}  \nonumber \\
    & \hspace{0.2in}+\int \log\left(\frac{\prod_{2M+1}^{K} P_{q,1}(\tilde{x}^k)  }{\prod_{2M+1}^{K} P_{j,1}(\tilde{x}^k)}\right) \tilde{P}_{q}(\tilde{\mathbf{x}}) d\tilde{\mathbf{x}}  \nonumber \\
    &\overset{(d)}{\geq} \int \log\left(\frac{\prod_{2M+1}^{K} P_{q,1}(\tilde{x}^k)  }{\prod_{2M+1}^{K} P_{j,1}(\tilde{x}^k)}\right) \prod_{2M+1}^K P_{q,1}(\tilde{x}^k) d\tilde{x}^{2M+1}\ldots d\tilde{x}^{K}\nonumber \\
    & = \sum_{k=2M+1}^K\int \log\left(\frac{P_{q,1}(x^k)  }{ P_{j,1}(x^k)}\right) P_{q,1}(x^k) dx^k\nonumber \\
    &\overset{(e)}{=}  (|\cN|-M) I(q,j),
\end{align}
where (a) follows from the same steps reaching \eqref{eqn:cusum_to_mcusum}, that is,
\begin{align} \label{eq_equi_transformedPDF}
    \sum_{\mathbf{x}:\tau_{2M}(\mathbf{x})=\tilde{\mathbf{x}}} P_{q,s}(\mathbf{x}) = \sum_{s\in\mathcal{S}}P_{q,s}(\tilde{\mathbf{x}})\frac{2M!}{\binom{2M}{M}};
\end{align}
(b) is because of the independence in \eqref{eq:til_P_s}; (c) holds by $\tilde{P}_{q}(\tilde{\mathbf{x}})$ equals to \eqref{eq_equi_transformedPDF} and marginalizing $\tilde{x}^{2M+1},\ldots, \tilde{x}^{K}$ out for the first integration, where we define a new PDF
\begin{equation}
    \hat{P}_{q,s}(\tilde{x}^1,\ldots,\tilde{x}^{2M}) = \sum_{s\in\mathcal{S}} \prod_{k'=1}^{2M} P_{q,s(k')}(\tilde{x}^{k'})\frac{2M!}{\binom{2M}{M}};
\end{equation}
(d) follows by noting that the first integration in (c) is the KL divergence between $\hat{P}_{q,s}$ and $\hat{P}_{j,s}$, which is always non-negative \cite{CT91}, together with marginalizing $\tilde{x}^{1},\ldots, \tilde{x}^{2M}$ in the second integration and the fact that the sensors in the third group are always honest in \eqref{eq:til_P} (see \eqref{eq:til_P_s}); and (e) follows from the selection of reverse attack in Section \ref{sec:mul_rev_attack} and definition \eqref{eq_KL_eachsensor}. Now, recall $q_m = \arg_q \min I_q$. For the pair $(q_m,j^*_{q_m})$ defined in Section \ref{sec:mul_rev_attack}, besides the inequality in \eqref{eq:til_I_qj}, we can further show the following equality
\begin{align}\label{eq:til_I_qj*}
    &\tilde{I}(q_m,j^*_{q_m})
    = \int \log\left(\frac{\sum_{s\in\mathcal{S}} P_{q_m,s}(\tilde{\mathbf{x}})}{\sum_{s\in\mathcal{S}} P_{j^*_{q_m},s}(\tilde{\mathbf{x}})}\right) \tilde{P}_{q_m}(\tilde{\mathbf{x}}) d\tilde{\mathbf{x}} \nonumber \\
    &\overset{(a)}{=} \int \log\left(\frac{\prod_{2M+1}^{K} P_{q_m,1}(\tilde{x}^k)  }{\prod_{2M+1}^{K} P_{j^*_{q_m},1}(\tilde{x}^k)}\right)\prod_{2M+1}^K P_{q,1}(\tilde{x}^k) d\tilde{x}^{2M+1}\ldots d\tilde{x}^{K}\nonumber \\
    &=\sum_{k=2M+1}^K\int \log\left(\frac{P_{q_m,1}(x^k)  }{ P_{j^*_{q_m},1}(x^k)}\right) P_{q_m,1}(x^k) dx^k\nonumber \\
    &= (|\cN|-M) I(q_m,j^*_{q_m}) = (|\cN|-M) \min_{q \in [Q]} I_q,
\end{align}
where equality (a) follows from the symmetry enforced in the second case of the reverse attack in \eqref{eq_mul_rev_attacksym} and the steps for reaching \eqref{eq_mixtureCUSUM1}. More specifically, in (\ref{eq:til_I_qj} b)
\[
\sum_{s\in\mathcal{S}} \prod_{k'=1}^{2M} P_{q_m,s(k')}(\tilde{x}^{k'}) =  \sum_{s\in\mathcal{S}} \prod_{k'=1}^{2M} P_{j^*_{q_m},s(k')}(\tilde{x}^{k'})
\]
from $P_{q_m,1}=P_{q_m}=P_{j^*_{q_m},2}$ and $P_{q_m,2}=P_{j_{q_m}^*}=P_{j^*_{q_m},1}$.} Now plugging \eqref{eq:til_I_qj} and \eqref{eq:til_I_qj*} into \eqref{eq:til_I*} shows that
\begin{equation}
    \tilde{I}^* = (|\cN|-M) I^*,
\end{equation}
since $(|\cN|-M) I^* \leq (|\cN|-M) I(q,j) \leq \tilde{I}(q,j), \forall q,j$. Noting that $\mathcal{D}[T]\geq \mathcal{D}_{\textrm{trans}}[T^*_{\mathrm{QCD}}]\sim\frac{\log\gamma}{(|\cN|-M) I^*}$ under $\mathcal{A}[T] \geq r$, as $\gamma\rightarrow\infty$, it completes the proof of the converse part.

\subsection{Proof of the acheivability for multi-hypothesis BDQCD in Theorem \ref{thm:scaling_M_achiv}}\label{subsec:stop_rule_M}
Now, we first describe the local decision rule at each honest sensor, and then propose two global fault-tolerant decision rules for the acheivability in part (a) and (b) of Theorem \ref{thm:scaling_M_achiv} respectively.
\subsubsection{Local decision rule: ``Soft" Matrix CUSUM}
Since the honest sensors are not allowed to cooperate with each other, it is natural to adopt the matrix CUSUM algorithm reviewed in Section \ref{subsec:M_BDQCD}. Note that for the original matrix CUSUM in \cite{TV02}, since there is only one honest node, it makes perfect sense for the procedure to make a hard decision and terminate after the alarm; however, in our setting, the task is not done yet until the fusion center has determined the occurrence of the event. Therefore, we adapt the matrix CUSUM procedure to the ``soft" version as follows. Whenever a $Y^k_{t,q}$ under \eqref{eq_matrix_CUSUM_matirx} exceeds the threshold $h$ at time index $t$, the hypothesis $H_{q}$ is {\it softly} decided by informing the fusion center that this hypothesis is {\it acceptable} at the sensor $k$. Now each honest sensor may keep monitoring the event and report multiple hypotheses to the fusion center. Later in Sec. \ref{subsec_fusionrules_analy}, this soft version will help us resolve the ``undecidable event", which disables the fusion center to make a conclusive decision.

Formally, for the soft matrix CUSUM procedure, a hypothesis $H_{q}$ is acceptable by the node $k$ at time
\begin{equation}  \label{eq_MatrixCUSUMacceptabletime}
    \sigma_k^{q}(h) := \inf \left\{ t\in\mathbb{N}: Y^k_{t,q}\geq h \right\}.
\end{equation}
In contrast, for the original matrix CUSUM \cite{TV02}, a hypothesis $H_q$ is hard decided at time $\sigma_k^{q}(h)$ if $\sigma_k^{q}(h)$ equals to
\[
 \sigma_k(h) : =  \min_{\hat{q} \in [Q]}\sigma_k^{\hat{q}}(h) \; \mbox{and} \; q=\arg \max_{\hat{q} \in [Q]} \left( Y^k_{t,\hat{q}}|_{t=\sigma_k(h)} \right).
\]

\subsubsection{Global fault-tolerant decision rules} \label{subsec_fusionrules}
As a baseline, the one-shot rule which uses the original matrix CUSUM is first introduced
as \\
\noindent One-shot $d$-th alarm: This family of rules is a direct extension of the one-shot rule for the binary case in \cite{BFL16}, \cite{BF16} to the multi-hypothesis setting. Each sensor adopts the original matrix CUSUM \cite{TV02} as its local report mechanism and reports the first acceptable non-zero hypothesis as soon as the sensor finds it. The fusion center declares that an abrupt event has occurred at the first time that a hypothesis, say $H_q$, has received $d$ local reports. It also declares that the hypothesis $H_q$ is true. 

Now, we propose two rules based on the ``soft" matrix CUSUM. \\
\noindent {\bf i) Multi-shot $d$-th alarm $\tau^m_{(d)}(h)$}: This family of rules requires each sensor to adopt the soft version of matrix CUSUM and to alarm whenever a hypothesis $H_{\hat{q}}$, $\hat{q}\in [Q]$, is acceptable. Formally, for each $k\in \cN$, sensor $k$ reports $H_{\hat{q}}$ at the time index $\sigma_k^{\hat{q}}(h)$, for every $\hat{q}\in[Q]$. In this reporting mechanism, we stipulate that for each sensor, every hypothesis can be reported at most once, and a reported hypothesis cannot be withdrawn. In other words, once reported by a sensor, a hypothesis will be promoted as a candidate by that sensor ever since. If a tie happens at an honest node $k$, then all the hypothesis indexes have the same $\sigma_k^{\hat{q}}(h)$ will be reported one after another, starting from the one with the largest $Y_{t,\hat{q}}^k$. For the case where two or more hypotheses have the same $Y_{t,\hat{q}}^k$, we break the tie randomly. Consecutive ties and/or multi-way ties can be easily resolved by equipping each node with a queue of size $Q-1$ and clearing the queue on the first-come first-serve basis. The fusion center declares that an abrupt event has occurred at the first time that a hypothesis, say $H_q$, has been deemed acceptable by $d$ sensors. It also declares that the hypothesis $H_q$ is true.\\
\noindent {\bf ii) Simultaneous $d$-th alarm $\tau^s_{(d)}(h)$}: Each sensor constantly transmits $Q$ bits local decision at time index $t$ to indicate whether $H_{\hat{q}}$ is acceptable, $\forall \hat{q} \in [Q]$. The fusion center declares that an abrupt event of type $q$ has occurred at the first time that a hypothesis, say $H_q$, has been simultaneously accepted by no less than $d$ sensors. 

We note that the three families of rules have different bandwidth and/or energy requirements. The one-shot scheme is the most bandwidth- and energy-efficient one as it requires each link to support $\lceil\log_2Q\rceil$ bits and this link is used only once. The multi-shot scheme also requires links to support $\lceil\log_2Q\rceil$ bits, but each link may be used up to $Q$ times. As for the simultaneous rule, it requires each link to support $Q$ bits and each link is constantly used. Also, it is worth noting that the soft matrix CUSUM reduces to the original CUSUM adopted in \cite{BFL16} when $Q=1$, i.e. binary hypothesis. Thus, the proposed multi-shot and simultaneous $d$-th alarm include the one-shot and voting rules in \cite{BFL16}, \cite{BF16} as special cases, respectively. Moreover, depending on the application at hand, the fusion center can opt to stop only once or multiple times. When the fusion center chooses to stop only once as \cite{BFL16}, $T_\rho=\infty$ in \eqref{eq:stopping_rule_multi} for $\rho>1$. For the scenario where the fusion center makes multiple alarms, it restarts with the same global decision rule after each alarm at $T_\rho$. Our upcoming Proposition \ref{Thm_multi_finite_h} applies to both scenarios.

\subsubsection{Performance analysis} \label{subsec_fusionrules_analy}
We now carry out the worst-case analysis on the performance of the multi-shot $d$-th alarm and simultaneous $d$-th alarm rules. For the one-shot $d$-th alarm, we point out a notable difference from the binary counterpart \cite{BFL16} which significantly degrade the performance from the converse in Theorem \ref{thm:scaling_M}. The \textit{undecidable event} may happen: it is possible that there is no non-zero hypothesis index with enough local alarms for making a decision, even though all honest sensors have raised alarms; thereby, the detection delay is infinity.
Unfortunately, even more advanced multi-shot $d$-th alarm can only achieve the converse when $|\cN|=M+1$ from the upcoming analysis. When $|\cN|=M+1,$ if the one-shot $d$-th alarm is used, the compromised sensors can easily trigger the undecidable event.

To further characterize  \eqref{eq:ess_delay_multiple} and
\eqref{eq:falseIsolationTime_multiple} for the two proposed rules,
we will prove asymptotic dominance results in upcoming Lemma
\ref{lem:sigmas equal}, which greatly simplifies the delay
analysis. Intuitively, although there are multiple sensors and
$Q+1$ hypotheses, for each honest sensor being considered, we only
have to examine the statistics between the $q$-th hypothesis and
the one that is ``\textit{closest}" to $q$, for every $q \in [Q]$.
Note that this intuition also comply with our asymptotic converse.
Before introducing Lemma \ref{lem:sigmas equal} and the complete
analysis, some definitions and a proposition regarding the
compromised sensors will be provided first. Similarly to
\cite{BFL16}, from \eqref{eq_MatrixCUSUMacceptabletime}, we define
the ordered time indexes
$\sigma^q_{(1)}(h)\leq\ldots\leq\sigma^q_{(|\cN|)}(h)$,  for all
$q \in [Q]$, over $|\cN|$ honest sensors as if there is no
compromised sensor, for softly deciding hypothesis $H_q$ (cf. the
$q$th row of CUSUM matrix \eqref{eq_matrix_CUSUM_matirx})
We also let $S_\ell^q(h)$ be the first time that the hypothesis $H_q$ is simultaneously softly-decided by $\ell$ honest sensors, defined as
\begin{align} \label{eq_simuMatrixCUSUMacceptabletime}
\inf\!\left\{t\in\N \!: \! Y^{k}_{t,q} \ge h\ \forall k\in \! \mathcal{L},\ \!\!\mbox{for}~\mbox{some}\!\ \mathcal{L} \! \subset \! [\cN],\ \!\!|\mathcal{L}|=\ell\right\}.
\end{align}
Finally, we will use $\Eqv$ to represent the expectation when the change with hypothesis index $q$ happens at time $\nu$ and the compromised sensors are absent.

To continue the worst-case analysis in  \eqref{eq:ess_delay_multiple} and \eqref{eq:falseIsolationTime_multiple}, recall that all the compromised sensors know the actual $\nu$ and the actual hypothesis $q$. They can then collaboratively attack/confuse the fusion center. Thus, it is obvious that choosing any $d \leq M$ is bad for false alarm or false isolation in \eqref{eq:falseIsolationTime_multiple}, while any $d > |\cN|$ is bad for detection delay in \eqref{eq:ess_delay_multiple}. We therefore confine the choice of $d$ to some reasonable region and obtain the following result.

\begin{proposition} \label{Thm_multi_finite_h}
Fix $h>0$. For any positive integer $Q$,
and $d\in\{M+1,...,|\cN|\}$, for multi-shot $d$-th alarm, we have
\begin{align}
\A[\tau^m_{(d)}(h)] &\geq \min_{q\in[Q]^+} \min_{\hat{q}\in [Q]\setminus\{q\}}\Eq[\sigma^{\hat{q}}_{(d-M)}(h)], \label{eq_MultishotWorstI}\\
\cD[\tau^m_{(d)}(h)] &\leq  \max_{q} \Eq[\sigma^q_{(d)}(h)]+Q-1; \label{eq_MultishotWorst}
\end{align}
while for simultaneous $d$-th alarm
\begin{align}
\A[\tau^s_{(d)}(h)] &\geq\min_{q\in[Q]^+}\min_{\hat{q}\in[Q]\setminus\{q\}} \Eq[S^{\hat{q}}_{d-M}(h)]; \label{eq_SimutWorstI}\\
\cD[\tau^s_{(d)}(h)] &\leq  \max_{q} \Eq[S^q_d(h)]. \label{eq_SimutWorst}
\end{align}
\end{proposition}
\begin{IEEEproof}
    See Appendix~\ref{apx:proof_prop}.
\end{IEEEproof}


Based on Proposition \ref{Thm_multi_finite_h}, in what follows, we provide explicit upper bounds on the detection delay. Lemma \ref{lem:sigmas equal} is shown first, which states that for each hypothesis $q\in[Q]$, although there are total $Q+1$ hypotheses, one only has to worry about the one that is ``\textit{closest}" to $q$ in terms of the KL divergence.
\noindent By writing $\P_q$ for $\mathsf{P}_{0}^{q,g=\emptyset}$ and recall the definition of $j^*_q$ in Assumption~\ref{asm:I}, we have:
\begin{lemma}\label{lem:sigmas equal}
Suppose $h$ is large enough and Assumption~\ref{asm:I} holds. For
any $q\in[Q]$, it holds $\P_q$-a.s. that \\
(i) The first time $H_{q}$ is softly decided at the honest sensor $k$,
$\sigma^q_k(h)$ in \eqref{eq_MatrixCUSUMacceptabletime}, equals to
\begin{equation}\label{eq_closet_accept_time}
\sigma^{q,j^*_q}_k(h):= \inf\{t\in\N : Y^k_{t}(q,j^*_q)\ge h\}.
\end{equation}
(ii) For any $|\cN| \geq d \geq 1$, the first time $H_{q}$ is
simultaneously softly-decided by $d$ honest sensors, $S^q_d(h)$,
equals to
\begin{equation}
S^{q,j^*_q}_d(h):= \inf\left\{t\in\N :
Y^{(|\cN|-d+1)}_{t}(q,j^*_q)\ge h\right\},
\end{equation}
where $Y^{(1)}_{t}(q,j^*_q) \leq \ldots \leq
Y^{(|\cN|)}_{t}(q,j^*_q)$ are the ordered CUSUM statistics of
$Y^k_{t}(q,j^*_q)$, for hypotheses $q$ and $j^*_q$ at time $t$.
\end{lemma}
\begin{IEEEproof}
    See Appendix~\ref{sec:app}.
\end{IEEEproof}

We are now ready to present the results on the asymptotic delay performance. Let $Z_{(1)}$, $Z_{(2)}$, \ldots, $Z_{(|\cN|)}$ be the order statistics of independent standard  normal random variables. For each $d\in\{1,2, \ldots,|\cN|\}$, we denote by $\xi_d$ the expected value of $Z_{(d)}$. Moreover, for each $q\in[Q]$, we set
$D^q_{d:|\cN|} := \xi_d \sqrt{\frac{\sigma^2(q,j^*_q)}{I^q}}$. 

\begin{theorem}\label{Thm_QaryDelay}
Suppose Assumption~\ref{asm:I} holds. As $h\to\infty$, for any $q \in [Q]$ and $1\leq d \leq |\cN|$, we have
\begin{equation}\label{eq_Qarydelayq}
  \Eq[\sigma^q_{(d)}(h)] = \frac{h}{I^q} + D^q_{d:|\cN|} \sqrt{h}(1+o(1)), 
\end{equation}
and the detection delay of the multi-shot $d$-th alarm in \eqref{eq_MultishotWorst} is upper-bounded as
\begin{equation}\label{eq_Qarydelay}
\cD[\tau^m_{(d)}(h)] \leq \max_{q} \left (\frac{h}{I^q} + D^q_{d:|\cN|} \sqrt{h}(1+o(1))\right).
\end{equation}
\end{theorem}
\begin{IEEEproof}
    See Appendix~\ref{apx:proof_delay}.
\end{IEEEproof}

\begin{theorem}\label{CoroSimu_Qarydelay}
Suppose Assumption~\ref{asm:I} holds. As $h\to\infty$, for any $q \in [Q]$ and $1\leq d \leq |\cN|$, we have
\begin{equation}\label{eqn:simul_delay0}
    \Eq\left[S^q_{d}(h)\right] \leq \frac{h}{I^q} + D^q_{d:|\cN|} \sqrt{h}(1+o(1)),
\end{equation}
and the detection delay of the simultaneous $d$-th alarm in \eqref{eq_SimutWorst} is upper-bounded as
\begin{equation}\label{eqn:simul_delay}
\cD[\tau^s_{(d)}(h)] \leq \max_{q} \left (\frac{h}{I^q} + D^q_{d:|\cN|} \sqrt{h}(1+o(1))\right).
\end{equation}
\end{theorem}
\begin{IEEEproof}
    See Appendix~\ref{apx:proof_delay}.
\end{IEEEproof}

The asymptotic performance of the mean time to false alarm/isolation of the proposed families of rules are given in the following. 
\begin{theorem} \label{Thm_one_QaryFalse}
Fix $h>0$. If $M<d\leq|\cN|$, the mean time to a false alarm or a false isolation $\A[\tau^m_{(d)}(h)]$ in \eqref{eq_MultishotWorstI} for the multi-shot $d$-th alarm is lower-bounded by
\begin{equation} \label{eq_QaryFalse}
 \frac{d-M}{(d-M+1)}\left(\begin{array}{c}
|\cN| \\ d-M \end{array} \right)^{\frac{-1}{d-M}}\exp(h).
\end{equation}
\end{theorem}
\begin{IEEEproof}
    See Appendix~\ref{apx:proof_false}.
\end{IEEEproof}

\begin{theorem} \label{CoroSimu_QaryFalse}
Fix $h>0$. If $ M< d \leq |\cN|$, the mean time to a false alarm or a false isolation $\A[\tau^s_{(d)}(h)]$ in \eqref{eq_SimutWorstI} is lower-bounded by
\begin{equation}\label{eqn:simul_false}
\frac{1}{2}\left(\begin{array}{c} |\cN| \\
d-M
\end{array}\right)^{-1}\exp\left((d-M)h\right).
\end{equation}
\end{theorem}
\begin{IEEEproof}
    See Appendix~\ref{apx:proof_false}.
\end{IEEEproof}


Although the delay upper bounds in Theorems~\ref{Thm_QaryDelay}
and~\ref{CoroSimu_Qarydelay} are identical, we will show the
superiority of the simultaneous rule when detection delay and mean
time to false alarm/isolation are jointly considered, which also
validates Theorem \ref{thm:scaling_M_achiv}. 
For simultaneous
$d$-th alarm $\tau^s_{(d)}(h)$, one can ensure
$\A[\tau^s_{(d)}](h) \geq \gamma$ from \eqref{eqn:simul_false}
by selecting local threshold
\begin{equation}\label{eqn:simul_h}
\frac{1}{d-M}\left(\log \gamma+\log\left(2\left(\begin{array}{c} |\cN| \\ d-M
\end{array} \right)\right)\right).
\end{equation}
Also by plugging $h$ in \eqref{eqn:simul_h} into
\eqref{eqn:simul_delay}, with $\gamma \rightarrow \infty$,
\begin{align}
    &\cD[\tau^s_{(d)}(h)]  \lesssim \max_q \left( \frac{\log \gamma}{(d-M)I^{q}} \right).
\end{align}
Then part (a) of Theorem \ref{thm:scaling_M_achiv} is valid since
$d=|\cN|$ minimizes the right hand side above. Moreover, for the
multi-shot $d$-th alarm $\tau^m_{(d)}(h)$, one can ensure
$\A[\tau^m_{(d)}(h)] \geq \gamma$ from \eqref{eq_QaryFalse} by
selecting local threshold
\begin{equation} \label{eq_multih1}
    h=\log \gamma+\frac{1}{d-M}\log\left(\begin{array}{c} |\cN| \\ d-M
\end{array} \right)+\log\left(\frac{d-M+1}{d-M}\right).
\end{equation}
Plugging $h$ in \eqref{eq_multih1} into \eqref{eq_Qarydelay} and
let $\gamma \rightarrow \infty$, we have
\begin{align}
  \cD[\tau^m_{(d)}(h)] &\lesssim  \max_q \left( \frac{\log
  \gamma}{I^{q}}\right).
\end{align}
This
validates part (b) of Theorem \ref{thm:scaling_M_achiv}.

\begin{remark}
Apart from being the building block of our proof, Lemma~\ref{lem:sigmas equal} also reveals another practical benefit. It basically confirms that for large $h$, for each row of the matrix CUSUM, i.e., each abnormal hypothesis $H_q$, an honest sensor only has to compute and update the CUSUM statistics $Y^k_t(q,j^*_q)$. This is particularly useful in applications where sensors are subject to stringent energy constraints. Moreover, although Lemma~\ref{lem:sigmas equal} only shows asymptotic optimality of this approach, we provide in Fig. \ref{Fig_reduced_Delay_single_meter} an example showing that the intuition of updating only the closest hypothesis also applies to small $h$.
In this example, we consider $Q=2$ with $P_0,P_1,P_2$ the PDFs of Gaussian random variables with means $0,1,-1$ and same variances $\sigma^2$, respectively. Instead of the original CUSUM matrix, we use the following reduced one
\begin{eqnarray} \label{eq_matrix_CUSUM_matirxreduce}
\begin{bmatrix}
  Y^k_t(1,0) & \infty  \\
  \infty & Y^k_t(2,0). \\
\end{bmatrix}
\end{eqnarray}
As shown in Fig. \ref{Fig_reduced_Delay_single_meter}, the detection performances of full CUSUM matrix and the reduced one are almost identical, as expected.
\begin{figure}[h]
    \centering
    \includegraphics[scale=0.5]{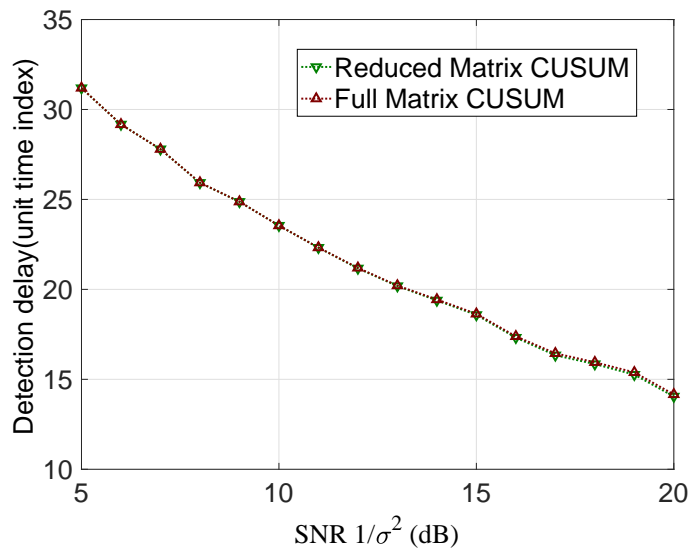}
    \caption{Detection delay of full Matrix CUSUM and that of the reduced one in \eqref{eq_matrix_CUSUM_matirxreduce} for an honest sensor. The local threshold is set such that $\exp(h) = 10^4$ ($h\approx 9.21$) and each curve is calculated based on 20000 realizations.}
     \label{Fig_reduced_Delay_single_meter}
\end{figure}
\end{remark}


\section{Game-theoretic formulation}\label{sec:game}
In this section, we formulate a leader-follower Stackelberg game \cite[Section 3.6]{basar1999dynamic} for the considered BDQCD, where the fusion center and honest sensors act as the leader and the compromised sensors act as the follower. It turns out that the characterization of the first-order optimality in the previous sections will help us characterize the Stacklberg equilibrium.

In the game, the information available to the two players are as follows:
\begin{itemize}
  \item The follower knows the leader's strategy $g_1$, all the current and past local observations $\mathbf{X}_1^t$, the change time $\nu$, and the actual hypothesis.
  \item The leader is oblivious of the exact indexes of compromised sensors, but knows the maximum number of compromised sensors $M$. 
\end{itemize}
Now we define the strategy spaces of the two players. With $Q+1$
hypotheses, we assume the noiseless link of each sensor can support (at least) $Q$ bits and use the $K\times 1$ vector
$\hat{\boldsymbol\lambda}_t\in \mathbb{Z}_{2^Q}^K$ to denote the
$Q$-bit local decisions at time $t$, where $\mathbb{Z}_{2^Q}$ is the integer ring modulo $2^Q$. The compromised sensors
cooperatively form the attack vector
$\mathbf{e}_t\in \mathbb{Z}_{2^Q}^K$ which has at most $M$ non-zero
components, reflecting that there are at most $M$ compromised
sensors. At time $t$, the fusion center receives
\begin{equation} \label{eq_receiveFusion}
     \hat{\boldsymbol\lambda}_t + \mathbf{e}_t,
\end{equation}
where the addition is over $\mathbb{Z}_{2^Q}$. A strategy $g_1$ of the leader at time $t$ includes a local decision rule that maps $X^k_1,\ldots,X^k_t$ into the $k$th entry of $\hat{\boldsymbol\lambda}_t$ at each sensor $k \in [K]$ (the leader treats all $K$ sensors as honest), and a stopping rule at the fusion center which maps $K \times t$ matrix $[\hat{\boldsymbol\lambda}_1 + \mathbf{e}_1, \ldots, \hat{\boldsymbol\lambda}_t + \mathbf{e}_t]$ to a decision $\hat{q}_t \in [Q]^+$. An alarm is fired if $\hat{q}_t \neq 0$. The stopping time for type $q\in[Q]$ is given in \eqref{eq:stopping_rule_multi} and again let $T$ be the first alarm time. A strategy $g_2$ of the follower at time $t$ is the vector $\mathbf{e}_t$ in \eqref{eq_receiveFusion}, where all $K$ elements are from $\mathbb{Z}_{2^Q}$ but elements in a subset of indices with size $|\cN|$ (corresponding to indexes of honest sensors) is deterministically 0. We use $\cG_1$ and $\cG_2$ to denote the pure-strategy spaces of aforementioned $g_1$ and $g_2$, respectively.

We first focus on the binary case as in Section \ref{sec:B_BDQCD}, i.e., $Q=1$, and refer  to the game as the {\it binary BDQCD Stackelberg game}. From the detection delay in \eqref{eq:delay_binary} and false alarm in \eqref{eq:falseAlarmTime_multiple}, we define the corresponding
performance metrics under strategies $(g_1,g_2)$ as
\begin{align}
    \cD(g_1,g_2) &:= \sup_{\nu} \esssup \E^{g_2}_{\nu}[ (T-\nu)^+ |\mathbf{X}_1^{\nu}], \quad\text{and}\notag\\
\A(g_1,g_2) &:= \E^{g_2}_{\infty}[ T],
\end{align}
respectively, where we note that the stopping time $T$ is a function of $(g_1,g_2)$. Since we wish the mean time to false alarm to be larger than a given $\gamma$, the cost for the leader is then defined as
\begin{equation} \label{eq_multigame_cost_leader}
 J^1(g_1,g_2) \triangleq \lim_{\gamma \rightarrow \infty} \left(\frac{\cD(g_1,g_2)}{\log \gamma}+ I_{-}\left(\gamma-\mathcal{A}(g_1,g_2)\right)\right),
\end{equation}
where function $I_{-}(u)$, as defined in \cite[Section 11.2]{BookBoyd}, is $\infty$ when $u>0$ and zero otherwise. Note that the change point $\nu$ is known at the follower, and thus $g_2$ can be different before and after the change time $\nu$. The cost for the follower is the mean time to false alarm
$
 J^2(g_1,g_2)\triangleq\mathcal{A}(g_1,g_2)
$
as the follower wants to sabotage detection by making false alarm more frequent.

For our game, the Stackelberg equilibrium strategy for the leader
and cost are defined as follows. As \cite[Definition
4.1]{basar1999dynamic}, we define
\begin{definition}
Fix $\eps>0$. For any $g_1 \in \cG_1$, the set $R_\eps^2(g_1)
\subseteq \cG_2$ defined by
\begin{equation}\label{eq_game_response}
R_\eps^2(g_1)=\left\{g_2 \in \cG_2 : J^2(g_1, g_2) \leq \inf_{\xi \in \cG_2} J^2(g_1,\xi)+\eps\right\}
\end{equation}
is the $\eps$-optimal response set of the follower to the strategy $g_1$ of the leader.
\end{definition}
\begin{definition} \label{Def_leaderCostSequence}
The Stackelberg cost of the leader is defined as
\begin{equation}\label{eq_leaderCostSequence}
J^{1*}=\inf_{g_1 \in \cG_1} \inf_{\eps>0} \sup_{g_2 \in R^2_{\eps}(g_1)}J^1(g_1,g_2).
\end{equation}
For any $\eps>0$, a
strategy $g^{1*}_\eps \in \cG_1$ is an $\eps$-Stackelberg
equilibrium strategy for the leader if
\[
\inf_{\bar{\eps}>0} \sup_{g_2 \in R_{\bar{\eps}}^2(g^{1*}_\eps)}J^1(g^{1*}_\eps,g_2) \leq
J^{1*}+\eps
\]
\end{definition}

Based on the above definitions, we first prove the following lemma which results in the game solution later in Theorem \ref{Thm_Binarygame_cost}.
\begin{lemma}\label{lma:scaling_lower_bound}
    For the binary BDQCD Stackelberg game, if there exists a pure strategy $\acute{g}_2\in\mathcal{G}_2$ that results in a lower bound $J^1(g_1,\acute{g}_2)\geq \eta$ for any $g_1\in\mathcal{G}_1$, then
    \begin{equation}
    \sup_{g_2 \in R^2_{\eps}(g_1)}J^1(g_1,g_2)\geq J^1(g_1,\acute{g}_2)\geq\eta
    \end{equation}
    for any $\eps>0$.
\end{lemma}
\begin{proof}
We prove that $\sup_{g_2 \in R_\eps^2(g_1)}J^1(g_1,g_2)\geq
J^1(g_1,\acute{g}_2)$. If $\acute{g}_2 \in  R_\eps^2(g_1)$, then
the inequality is trivial. If not, for any $g_2 \in R_\eps^2(g_1)$
we have $\mathcal{A}(g_1,g_2) \leq
\mathcal{A}(g_1,\acute{g}_2)$. If
$\mathcal{A}(g_1,\acute{g}_2)<\gamma$, then
$J^1(g_1,g_2)=J^1(g_1,\acute{g}_2)=\infty$ from
\eqref{eq_multigame_cost_leader}. Now consider
$\mathcal{A}(g_1,\acute{g}_2) \geq \gamma$. If
$\mathcal{A}(g_1,g_2) < \gamma  \leq
\mathcal{A}(g_1,\acute{g}_2)$, from
\eqref{eq_multigame_cost_leader},
\[
J^1(g_1,\acute{g}_2)\leq J^1(g_1,g_2) =\infty.
\]
Otherwise, if $\gamma \leq \mathcal{A}(g_1,g_2)  \leq \mathcal{A}(g_1,\acute{g}_2)$, then
\[
I_{-}\left(\gamma-\mathcal{A}(g_1,\acute{g}_2)\right)=I_{-}\left(\gamma-\mathcal{A}(g_1,g_2)\right)=0,
\]
we construct an attack $\hat{g}_2$ acting as $g_2$ when $\nu = \infty$ and as $\acute{g}_2$ otherwise. This $\hat{g}_2$
will lie in $R_\eps^2(g_1)$ and result in $J^1(g_1,\hat{g}_2)=
J^1(g_1,\acute{g}_2)$, which results in
$J^1(g_1,\acute{g}_2)=J^1(g_1,\hat{g}_2)\leq \sup_{g_2 \in
R_\eps^2(g_1)}J^1(g_1,g_2)$.
\end{proof}

\begin{theorem} \label{Thm_Binarygame_cost}
  For the binary BDQCD Stackelberg game, the Stackelberg cost $J^{1*}$ is $\frac{1}{(|\cN|-M)I}$ when $|\cN|>M$ and zero
  elsewhere; and for any $\eps>0$, the $|\cN|$-voting rule is the $\eps$-Stackelberg
  equilibrium strategy for the leader.
\end{theorem}

\begin{proof}
From Lemma~\ref{lma:scaling_lower_bound}, for the converse $J^{1*} \geq \frac{1}{(|\cN|-M)I}$ it suffices to construct an attack $\acute{g}_2$ such that for any $g_1\in\mathcal{G}_1$, we have
\begin{equation} \label{eq_binary_limitedBW_reverse_attack}
    J^1(g_1,\acute{g}_2)\geq \frac{1}{(|\cN|-M)I},
\end{equation}
when $|\cN|>M$. This is valid from the proof of Theorem
\ref{thm:new_converse} by choosing $\acute{g}_2$ as the proposed
reverse attack. On the other hand, the achievability comes from
\eqref{eqn:opt_scale}. That is, the $|\cN|$-voting rule, which
uses only 1 bit from each sensor, achieves
\begin{equation} \label{eq_binary_achieve}
 \max_{g_2 \in R^2_\eps\left(\tau^s_{(|\cN|)}(h)\right)} J^1\left(\tau^s_{(|\cN|)}(h),g_2\right) = \frac{1}{(|\cN|-M)I},\ \ \forall
 \eps>0,
\end{equation}
by selecting local threshold $h$ for the worst case attack (where all compromised sensors send ``1" always)  such that
\begin{equation}
\mathcal{A}(\tau^s_{(|\cN|)}(h),g_2) \geq \gamma, \; \; \forall g_2 \in
R^2_\eps\left(\tau^s_{(|\cN|)}(h)\right).
\end{equation}
Then $J^{1*} \leq \frac{1}{(|\cN|-M)I}$ and it concludes the
proof.
\end{proof}

We now consider $Q>1$ and define the multi-hypothesis BDQCD Stackelberg game. From the detection delay defined in \eqref{eq:ess_delay_multiple} and false alarm/isolation defined in
\eqref{eq:falseIsolationTime_multiple}, we define the corresponding performance metrics under strategy $(g_1,g_2)$ as
\begin{align}
    \cD(g_1,g_2) &:= \sup_{\nu}\sup_{q \in [Q]}\esssup \mathbb{E}^{q,g_2}_{\nu}[
    (T-\nu)^+|\mathbf{X}_1^{\nu}].  \\
\A(g_1,g_2) &:= \inf_{q \in [Q]^+}
\inf_{\hat{q}\in[Q]\setminus\{q\}}\mathbb{E}^{q,g_2}_{0}[
T^{\hat{q}} ].
\end{align}
Unfortunately in this case, we are unable to prove results similar to Lemma \ref{lma:scaling_lower_bound} with the follower's cost $J^2(g_1,g_2)=\mathcal{A}(g_1,g_2)$. The main difficulty is that when $\nu=0$, one needs to consider  jointly  the mean time to a false alarm and that to a false isolation, when it comes to constructing an attack $\acute{g}_2$. Therefore, we instead define $J^2(g_1,g_2)=-J^1(g_1,g_2)$ and show the following result.
\begin{theorem} \label{Thm_multigame_cost}
For multi-hypothesis BDQCD Stackelberg game, the Stackelberg
cost is $J^{1*}=\frac{1}{(|\cN|-M)I^*}$ when $|\cN|>M$ and
zero elsewhere; and for any $\eps>0$, the simultaneous
$|\cN|$-th alarm is the $\eps$-Stackelberg
  equilibrium strategy for the leader.
\end{theorem}
\begin{proof}
In this case, $J^{1*}$ in \eqref{eq_leaderCostSequence} becomes
\begin{equation}
\inf_{g_1 \in \cG_1}\sup_{g_2 \in \cG_2} J^1(g_1,g_2)
\end{equation}
If we can construct an attack $\acute{g}_2$ such that for any $g_1\in\mathcal{G}_1$, $J^1(g_1,\acute{g}_2)\geq \frac{1}{(|\cN|-M)I}$, and then $\sup_{g_2 \in \cG_2} J^1(g_1,g_2) \geq J^1(g_1,\acute{g}_2)\geq \frac{1}{(|\cN|-M)I}$ by definition. The game solution simply follows from the proof of Theorem \ref{thm:scaling_M} and \ref{thm:scaling_M_achiv} (a) by choosing $\acute{g}_2$ as our reverse attack for the multi-hypothesis case in Section \ref{sec:mul_rev_attack}.
\end{proof}

\section{Conclusions}\label{sec:conclude}
In this paper, the problem of BDQCD has been studied, where a
fusion center sequentially monitors an abrupt event via
distributed sensors which might be compromised. Both the binary
hypothesis and multi-hypothesis cases have been considered. For
the binary case, a novel converse bound for the first-order
asymptotic detection delay performance in the large mean time to a
false alarm regime has been proved. By comparing the converse
bound and the first-order scaling achieved by the existing
consensus rule, we have characterized the fundamental limit of
binary BDQCD in the large mean time to a false alarm regime (or
the small false alarm rate regime in a sense). For the
multi-hypothesis BDQCD, the novel converse has been generalized
from the binary case and the optimal first-order asymptotic
performance has again been characterized. Along with establishing
this fundamental result, two novel families of stopping rules have
been proposed, namely the multi-shot $d$-th alarm and the
simultaneous $d$-th alarm. The former is much more
energy-efficient and bandwidth efficient while the latter can
achieve asymptotically optimal performance under sufficient link
bandwidth whenever there are more honest sensors than compromised
ones. Finally, a leader-follower Stackelberg game has been
formulated based on the BDQCD problem discussed. The
asymptotically optimal stopping rule and the asymptotically worst
attack proposed for BDQCD have led us to the game solution, in
which the leader adopts the proposed asymptotically optimal
stopping rule (i.e., the simultaneous rule) and the follower
employs the corresponding asymptotically worst attack.

\appendices
\section{Lemmas}\label{apx:lemma}
In this appendix, some useful lemmas are presented and their proofs are given.

\begin{lemma}\label{lma:equivalence}
For any decision rule $T$, $\nu\ge 0$, $q\in[Q]$, and attack strategy $g$,
\begin{equation}
    \esssup \E^{q,g}_\nu[(T-\nu)^+\mid  \mathbf{X}_1^{\nu}] = \esssup \E^{q,g}_\nu[T-\nu\mid T>\nu,\ \mathbf{X}_1^{\nu}].
\end{equation}
\end{lemma}
\begin{IEEEproof}
For any decision rule $T$, $\nu\ge 0$, $q\in[Q]$, and attack strategy $g$, consider the subset $S$ of $\R^{K\times\nu}$ on which $T> \nu$. 
By the definition of $S$, observe that
\begin{equation}\label{prob}
\mathsf{P}^{q,g}_\nu[T>\nu \mid  \mathbf{X}_1^{\nu}] =
\begin{cases}
1,\quad &\hbox{on}\ S,\\
0,\quad &\hbox{on}\ S^c.
\end{cases}
\end{equation}
Also, note that
$\E^{q,g}_\nu[(T-\nu)^+\mid  \mathbf{X}_1^{\nu}]$
is well-defined on $\R^{K\times\nu}$, and by definition constantly zero on $S^c$. On the other hand,
\begin{equation}\label{2}
\E^{q,g}_\nu[T-\nu\mid T>\nu,\ \mathbf{X}_1^{\nu}]
\end{equation}
is well-defined only on $S$; specifically, since $\mathsf{P}^{q,g}_\nu[T>\nu \mid  \mathbf{X}_1^{\nu}] =0$ on $S^c$, the conditional probability needed to evaluate \eqref{2} is not well-defined on $S^c$.

Now, by the law of total expectation,
\begin{align}
\E^{q,g}_\nu &[(T-\nu)^+\mid  \mathbf{X}_1^{\nu}] \notag\\
&= \E^{q,g}_\nu[(T-\nu)^+\mid T>\nu,\ \mathbf{X}_1^{\nu}]\ \mathsf{P}^{q,g}_\nu[T>\nu\mid \mathbf{X}_1^{\nu}]\notag\\
&\hspace{0.2in} + \E^{q,g}_\nu[(T-\nu)^+\mid T\le \nu,\ \mathbf{X}_1^{\nu}]\ \mathsf{P}^{q,g}_\nu[T\le\nu\mid \mathbf{X}_1^{\nu}]\notag\\
&= \E^{q,g}_\nu[T-\nu\mid T>\nu,\ \mathbf{X}_1^{\nu}]\quad \hbox{on}\ S, \label{equal on S}
\end{align}
where the last equality follows from \eqref{prob}. Since $\E^{q,g}_\nu[(T-\nu)^+\mid  \mathbf{X}_1^{\nu}]$ is constantly zero on $S^c$, the desired result is a direct consequence of \eqref{equal on S}.
\end{IEEEproof}

\begin{lemma}\label{lem:general}
    In binary genie-aided BCQCD, for any general (not necessarily masked symmetric) fusion rule $T'(\{\mathbf{X}_t\}_{t \geq 1})$, there is a masked symmetric rule $T(\{\mathbf{X}_t\}_{t \geq 1})$ that is not worse than $T'(\{\mathbf{X}_t\}_{t \geq 1})$.
\end{lemma}
\begin{IEEEproof}
The proof is a constructive one. Recall we have defined $\pi:[K]\rightarrow[K]$ a masked permutation function that permutes the first $2M$ entries while keeps the remaining $|\cN|-M$ entries unchanged. Let $\Pi_{2M}$ be the collection of all ($2M!$ in total) such $\pi$. For $\mathbf{X}_t=[X_t^1,\ldots,X_t^K]$, we slightly abuse the notation to write $\pi(\mathbf{X}_t)=[X_t^{\pi(1)},\ldots,X_t^{\pi(K)}]$. Let
    \begin{equation} \label{eq_symmetric_rule}
        T(\{\mathbf{X}_t\}_{t \geq 1}) = \frac{1}{2M!}\sum_{\pi\in\Pi_{2M}}T'(\{\pi(\mathbf{X}_t)\}_{t \geq 1}).
    \end{equation}
    Following the proof of part 1 of \cite[Lemma 4.2]{IHsiang18}, one can show that $T(\{\mathbf{X}_t\}_{t \geq 1})$ is indeed a masked symmetric strategy. Now the detection delay $\cD_{\textrm{genie}}[T(\{\mathbf{X}_t\}_{t \geq 1})]$ is
    \[
    \sup_{s \in \mathcal{S},\nu} \esssup\E^{s}_{\nu}\left[  \left(\frac{1}{2M!}\sum_{\pi\in\Pi_{2M}}  T'(\{\pi(\mathbf{X}_t)\}_{t \geq 1})-\nu\right)^+ \bigg|\mathbf{X}_1^{\nu} \right].
    \]
    Thus $\cD_{\textrm{genie}}[T(\{\mathbf{X}_t\}_{t \geq 1})]$ is no longer than
    \begin{align}\label{eqn:sym_Delay}
       &\frac{1}{2M!} \sum_{\pi\in\Pi_{2M}} \sup_{s \in \mathcal{S},\nu} \esssup \E^{s}_{\nu}\left[ (T'(\{\pi(\mathbf{X}_t)\}_{t \geq 1})-\nu)^+ |\mathbf{X}_1^{\nu}\right]\nonumber \\
       &\overset{(a)}{=} \frac{1}{2M!} \sum_{\pi\in\Pi_{2M}} \sup_{s \in \mathcal{S},\nu} \esssup \E^{s\circ\pi^{-1}}_{\nu}\left[ (T'(\{\mathbf{X}_t\}_{t \geq 1})-\nu)^+ |\mathbf{X}_1^{\nu}\right]\nonumber \\
       &\overset{(b)}{=} \frac{1}{2M!} \sum_{\pi\in\Pi_{2M}} \sup_{s' \in \mathcal{S},\nu} \esssup \E^{s'}_{\nu}\left[ (T'(\{\mathbf{X}_t\}_{t \geq 1})-\nu)^+ |\mathbf{X}_1^{\nu}\right]\nonumber \\
       &= \frac{1}{2M!} \sum_{\pi\in\Pi_{2M}} \cD_{\textrm{genie}}[T'(\{\mathbf{X}_t\}_{t \geq 1})] = \cD_{\textrm{genie}}[T'(\{\mathbf{X}_t\}_{t \geq 1})].
    \end{align}
Note that essential supremum of the right-hand side of (a) is taken under the probability measure whose density is specified by \eqref{eqn:prod_distribution} under $\theta=0$ and the compromised group assignment $s \circ \pi^{-1}$. Then (a) can be proved similar to (\ref{eq_no_infProd}a) by the fact
\begin{align}
& \E^s_\nu \left[\left(T'(\{\pi(\mathbf{X}_t)\}_{t\ge 1})-\nu\right)^+ \ \middle| \ \mathbf{X}_1^\nu\right] (x) =\notag \\
& \E^{s\circ \pi^{-1}}_\nu \left[\left(T'(\{\mathbf{X}_t\}_{t\ge 1})-\nu\right)^+ \ \middle| \ \mathbf{X}_1^\nu\right] (\pi(x)), \forall x\in \R^{K\times\nu} \notag
\end{align}
since the permutation $\pi(.)$ is one-to-one; and (b) is due to the fact that $\{s\circ\pi^{-1}|s\in\mathcal{S}\}=\mathcal{S}$. We can similarly show that for the mean time to false alarm of the new rule, $\A_{\textrm{genie}}[T(\{\mathbf{X}_t\}_{t \geq 1})] \geq   \A_{\textrm{genie}}[T'(\{\mathbf{X}_t\}_{t \geq 1})]$.
Then we conclude that the masked symmetric strategy $T(\{\mathbf{X}_t\}_{t \geq 1})$ is at least as good as $T'(\{\mathbf{X}_t\}_{t \geq 1})$.
\end{IEEEproof}

\begin{lemma}\label{lem:general_mult_false}
    In multiple-hypothesis genie-aided BCQCD, for any general (not necessarily masked symmetric) fusion rule $T'(\{\mathbf{X}_t\}_{t \geq 1})$, there is a masked symmetric rule $T(\{\mathbf{X}_t\}_{t \geq 1})$ that has longer mean time to a false alarm or a false isolation than $T'(\{\mathbf{X}_t\}_{t \geq 1})$.
\end{lemma}
\begin{IEEEproof}
Again, the symmetrized rule $T$ is formed as \eqref{eq_symmetric_rule} in the binary BCQCD. It is obvious that
\begin{align}
    &\A_{\textrm{genie}}[T] = \inf_{q\in[Q]^+}\inf_{s\in\mathcal{S}}\inf_{\hat{q}\in[Q]\setminus\{q\}}\mathbb{E}^{q,s}_{0}[ T^{\hat{q}}(\{\mathbf{X}_t\}_{t\geq 1}) ] \nonumber \\
    &= \inf_{q\in[Q]^+}\inf_{s\in\mathcal{S}}\inf_{\hat{q}\in[Q]\setminus\{q\}}\mathbb{E}^{q,s}_{0}\left[ \frac{1}{2M!}\sum_{\pi\in\Pi_{2M}} T'^{\hat{q}}(\{\pi(\mathbf{X}_t)\}_{t\geq 1}) \right] \nonumber \\
    &\geq \frac{1}{2M!}\sum_{\pi\in\Pi_{2M}} \inf_{q\in[Q]^+}\inf_{s\in\mathcal{S}}\inf_{\hat{q}\in[Q]\setminus\{q\}}\mathbb{E}^{q,s}_{0}\left[ T'^{\hat{q}}(\{\pi(\mathbf{X}_t)\}_{t\geq 1}) \right] \nonumber \\
    &=\frac{1}{2M!}\sum_{\pi\in\Pi_{2M}} \inf_{q\in[Q]^+}\inf_{s\in\mathcal{S}}\inf_{\hat{q}\in[Q]\setminus\{q\}}\mathbb{E}^{q,s\circ\pi^{-1}}_{0}\left[ T'^{\hat{q}}(\{\mathbf{X}_t\}_{t\geq 1}) \right] \nonumber \\
    &=\frac{1}{2M!}\sum_{\pi\in\Pi_{2M}} \inf_{q\in[Q]^+}\inf_{s'\in\mathcal{S}}\inf_{\hat{q}\in[Q]\setminus\{q\}}\mathbb{E}^{q,s'}_{0}\left[ T'^{\hat{q}}(\{\mathbf{X}_t\}_{t\geq 1}) \right] \nonumber \\
    &=\frac{1}{2M!}\sum_{\pi\in\Pi_{2M}} \A_{\textrm{genie}}[T'] = \A_{\textrm{genie}}[T'].
\end{align}
As a remark, the above proof generalizes that for the mean time to false alarm part for the binary BCQCD which is omitted in Lemma \ref{lem:general}.
\end{IEEEproof}

\begin{lemma}\label{lem:extend_Nikogonov}
    For any centralized multi-sensor and $(Q+1)$-hypothesis QCD rule $\tilde{T}$, subject to $\A[\tilde{T}] \geq \gamma$, the detection delay is lower-bounded by
    \[
     \mathcal{D}[\tilde{T}]\gtrsim \frac{\log\gamma}{\tilde{I}^*}, \; \mbox{as} \; \gamma \rightarrow \infty,
    \]
where $\tilde{I}^*$ is defined in \eqref{eq:til_I*}.
\end{lemma}

\begin{IEEEproof}
One can follow the same arguments in the proof of \cite[Theorem 2]{nikiforov95} to show this Lemma. Specifically, for any $q\in[Q]$,
take an arbitrary $\eps_q\in (0,1)$. We extend the sequence of
additional stopping variables $\tilde{T}_{a,0}:= 0 <\tilde{T}_{a,1}<\tilde{T}_{a,2}<...$  introduced in the beginning of the proof of \cite[Theorem 2]{nikiforov95} into multi-sensor version as
\begin{align*}
\tilde{T}_{a,i+1} &= \max_{q\in [Q]} \tilde{T}^q_{a,i+1},\\
\tilde{T}^q_{a,i+1} &= \inf\left\{n\ge \tilde{T}_i+1: \frac{\tilde{P}_q(\tilde{\mathbf{X}}_{\tilde{T}_i+1}) ... \tilde{P}_q(\tilde{\mathbf{X}}_{n})}{\tilde{P}_0(\tilde{\mathbf{X}}_{\tilde{T}_i+1}) ... \tilde{P}_0(\tilde{\mathbf{X}}_{n})}\le \eps_q\right\}.
\end{align*}
Then the rest of the proof simply follows \cite{nikiforov95}.
\end{IEEEproof}

\section{Proof of Proposition~\ref{Thm_multi_finite_h}}\label{apx:proof_prop}
To prove \eqref{eq_MultishotWorstI} and \eqref{eq_SimutWorstI}, we note that in the worst case, all the compromised sensors can raise alarms about the same hypothesis continuously. This implies that as soon as $d-M$, $d\in\{M+1,...,|\cN|\}$, honest sensors raise alarms of the same hypothesis, the compromised sensors can cooperatively enforce a false alarm event. If the fusion center stops only once, for the false isolation, it may declare the correct decision before the first $\sigma^{\hat{q}}_{(d-M)}(h)$ for $\tau^m_{(d)}(h)$ (or $S^{\hat{q}}_{d-M}(h)$ for $\tau^s_{(d)}(h)$) and then $T^{\hat{q}}=\infty$ in \eqref{eq:falseIsolationTime_multiple}, which results in a lower bound instead of equality in \eqref{eq_MultishotWorstI} (and \eqref{eq_SimutWorstI}). When the fusion center stops multiple times, recall that in both $\tau_{d}(h)$ and $\tau^s_{(d)}(h)$, the fusion center resets local CUSUM matrices of all sensors to the all zero matrix after each stop time. The mean of false alarm or isolation time $T^{\hat{q}}$ of $\tau_{d}(h)$ (respectively $\tau^s_{(d)}(h)$) is clearly lower bounded by that obtained by applying $\tau_{d}(h)$ (respectively $\tau^s_{(d)}(h)$) but without reset, which corresponds to the right hand side of \eqref{eq_MultishotWorstI} (respectively \eqref{eq_SimutWorstI}).


For the detection delay \eqref{eq:ess_delay_multiple}, the worst-case attack $\cG$ for both the proposed stopping algorithms happens when all compromised sensors always output local decisions corresponding to $H_0$. Moreover, we note that both the global decision rules mentioned above are non-decreasing functions in each entry of local CUSUM matirx $Y^k_t(q,j)$ in \eqref{eq_CUSUM_element} and the worst CUSUM statistic that the pre-change observations can impose is $Y^k_t(q,j)=0, t \leq \nu $. Hence, for the proposed algorithms, in the worst case, $T$ is not a function of previous observations and Lemma 3 in \cite{BFL16} can be applied to show the equivalence between \eqref{eq:ess_delay_multiple} and
\begin{equation}\label{eq:delay_multiple}
    \cD[T] = \sup_{q\in[Q]}\sup_g \Eqg[T],
\end{equation}
which corresponds to the scenario where the change occurs at $t=0$. We therefore only have to consider \label{eq:delay_multiple} as the worst-case expected detection delay in the sequel. For mulit-shot $\cD[\tau^m_{(d)}(h)]$ in \eqref{eq_MultishotWorst}, the fusion center has to wait for $d$ honest sensors accepting the true $H_q$. However, false isolation $q_f \neq q$ may still happens if $\sigma^{q'}_{(d)}(h)<\sigma^q_{(d)}(h)$, $q' \in [Q] \setminus \{q\}$. Moreover, the longest extra delay caused by ties is $Q-1$. Thus we have the upper-bound in \eqref{eq_MultishotWorst}. The upper-bound for simultaneous $\cD[\tau^s_{(d)}(h)]$ in \eqref{eq_SimutWorst} can be obtained similarly. 

\section{Proof of Lemma~\ref{lem:sigmas equal}}\label{sec:app}
Fix $q\in[Q]$ and $k\in \cN$. For any $j \in [Q]^+$ with $j\neq q$, the CUSUM statistics $Y^k_{t}(q,j)$ at sensor $k$ can be decomposed as $Y^k_t(q,j) = Z^k_t(q,j) +\xi^k_t(q,j)$, where
\[
Z^k_t(q,j) := \sum_{s=1}^t \log\left(\frac{P_q(X^k_s)}{P_j(X^k_s)}\right),\xi^k_t(q,j) := -\min_{0\le s<t} Z^k_s(q,j).
\]
Under $\P_q$, $Z^k(q,j)$ is a random walk with drift $I(q,j)>0$ and variance $\sigma^2(q,j)<\infty$. It follows that $Z^k_t(q)$, defined as
$(Z^{k}_t(q,1),\ldots, Z^{k}_t(q,q-1), Z^{k}_t(q,0), Z^{k}_t(q,q+1),\ldots, Z^{k}_t(q,Q)),$
is a $Q$-dimensional random walk. Also $Y^k_t(q)$, which is similarly defined as $Z^k_t(q)$ by replacing $Z^{k}_t(q,j)$ with $Y^{k}_t(q,j)$, is a $Q$-dimensional perturbed random walk, as discussed in \cite[Section 6.10]{Gut-book-09}\footnote{Note that while the exposition in \cite[Section 6.10]{Gut-book-09} focuses on two-dimensional perturbed random walks, the same results there can be generalized to multi-dimensional cases as stated in \cite[Remark 10.1, p. 208]{Gut-book-09}.}.

Now for any $j \in [Q]^+$ with $j\notin\{q,{j^*_q}\}$, at time index $\sigma^{q,{j^*_q}}_k(h)$ in \eqref{eq_closet_accept_time}, the CUSUM statistics for hypotheses $(q,j)$
\[
\frac{Y^k_{\sigma^{q,{j^*_q}}_k(h)}(q,j)}{h} \to \frac{I(q,j)}{I(q,{j^*_q})}=\frac{I(q,j)}{I^q}\quad \hbox{as $h\to\infty$},\quad \ \P_q\hbox{-a.s.},
\]
by \cite[Theorem 10.1, p.206]{Gut-book-09}. This, together with Assumption~\ref{asm:I}, implies that it holds $\P_q$-a.s. that $\forall j \in [Q]^+ \setminus  \{q,{j^*_q}\}$
\begin{equation}\label{>1}
\frac{Y^k_{\sigma^{q,{j^*_q}}_k(h)}(q,j)}{h}>1,
\end{equation}
as $h$ is large enough.
Now, observe that from \eqref{eq_MatrixCUSUMacceptabletime},
$\sigma^q_k(h) = \inf\left\{t\in\N : \min_{0\le j\le Q,\ j\neq q} Y^k_t(q,j) \ge h\right\},$
and the RHS equals to
\begin{align*}
 & \inf \! \left\{ \!t\in\N \! : \min_{0\le j\le Q,\ j\notin\{q,{j^*_q}\}} Y^k_t(q,j) \ge h\ \hbox{and}\ Y^k_t(q,{j^*_q})\ge h \! \right\}\\
& =  \sigma^{q,{j^*_q}}_k(h),\quad \hbox{as $h$ is large enough},\qquad \P_q\hbox{-a.s.},
\end{align*}
where the last line follows from \eqref{>1}. Since this relation is true for all $k\in\cN$ and $\cN$ is a finite set, we conclude that $\sigma^q_k(h)= \sigma^{q,{j^*_q}}_k(h)$ for all $k\in \cN$ as $h$ is large enough, $\P_q$-a.s. This concludes the proof for part (i).

For part (ii), from \eqref{eq_simuMatrixCUSUMacceptabletime}, $S^q_d(h)$ is equal to
\begin{align*}
   \inf\big\{t\in\N : \min_{0\le j\le Q,\ j\neq q} & Y^{k}_t(q,j) \ge h\ \ \forall k\in \mathcal{L},\ \notag \\
  & \hbox{for some}\ \mathcal{L}\subset [\cN],\ |\mathcal{L}|=d \big\}.
\end{align*}
Then from \eqref{>1}, $S^q_d(h)$ becomes
\begin{align*}
& \inf\!\left\{t\in\N \!: \! Y^{k}_t(q,{j^*_q}) \ge h\ \forall k\in \! \mathcal{L},\ \!\!\mbox{for}~\mbox{some}\!\ \mathcal{L} \! \subset \! [\cN],\ \!\!|\mathcal{L}|\!=\!\ell\right\}\\
&= \inf\left\{t\in\N : Y^{(K-d+1)}_t(q,{j^*_q})\ge h\right\}.
\end{align*}
Then as $h \rightarrow \infty$, $\P_q\hbox{-a.s.}$ we have $S^q_d(h)= S^{q,{j^*_q}}_d(h)$.

\section{Proofs of Theorems~\ref{Thm_QaryDelay} and \ref{CoroSimu_Qarydelay}}\label{apx:proof_delay}
We first prove Theorem~\ref{Thm_QaryDelay}.
\begin{IEEEproof}
Fix a $1\leq d\leq |\cN|$. 
From  \cite[Theorem 3.1]{BF16}, we know that as $h\to\infty$,
\begin{equation}
    \Eq[\sigma^{q,{j^*_q}}_{(d)}(h)] = \frac{h}{I^q} + D^q_{d:|\cN|} \sqrt{h}(1+o(1)).
\end{equation}
Since $\sigma^{q}_{(d)}(h)$ and $\sigma^{q,{j^*_q}}_{(d)}(h)$ are both nonnegative and non-decreasing in $h$, the monotone convergence theorem yields
\begin{equation}
    \lim_{h\to\infty} \Eq[\sigma^q_{(d)}(h)] = \Eq\left[\lim_{h\to\infty}\sigma^q_{(d)}(h)\right] = \Eq\left[\lim_{ h \to\infty}\sigma^{q,{j^*_q}}_{(d)}(h)\right] = \lim_{ h \to\infty}\Eq\left[\sigma^{q,{j^*_q}}_{(d)}(h)\right],
\end{equation}
where the second equality follows from Part (i) of Lemma \ref{lem:sigmas equal}. The above two equations together show \eqref{eq_Qarydelayq}. Finally, plugging \eqref{eq_Qarydelayq} into \eqref{eq_MultishotWorst} and observing that $Q-1$ vanishes as $h\rightarrow\infty$ results in \eqref{eq_Qarydelay}.
\end{IEEEproof}

We then provide a proof to Theorem~\ref{CoroSimu_Qarydelay}.
\begin{IEEEproof}
Fix a $1\leq d\leq|\cN|$. From \cite[Theorem 3.2]{BF16}, it follows that as $h\to\infty$,
\begin{equation}
    \Eq[S^{q,{j^*_q}}_{d}(h)] \le \frac{h}{I^q} + D^q_{d:|\cN|} \sqrt{h}(1+o(1)).
\end{equation}
Since $S^{q}_{d}(h)$ and $S^{q,{j^*_q}}_{d}(h)$ are both nonnegative and non-decreasing in $h$, the monotone convergence theorem yields
\begin{equation}
    \lim_{h\to\infty} \Eq[ S^q_{d}(h)] = \Eq\left[\lim_{h\to\infty} S^q_{d}(h)\right] = \Eq\left[\lim_{ h \to\infty} S^{q,{j^*_q}}_{d}(h)\right] = \lim_{ h \to\infty}\Eq\left[ S^{q,{j^*_q}}_{d}(h)\right],
\end{equation}
where the second equality follows from Part (ii) of Lemma~\ref{lem:sigmas equal}. The previous two equations together then give \eqref{eqn:simul_delay0}. Finally, plugging \eqref{eqn:simul_delay0} into \eqref{eq_SimutWorst} results in \eqref{eqn:simul_delay}. 
\end{IEEEproof}

\section{Proofs of Theorems~\ref{Thm_one_QaryFalse} and \ref{CoroSimu_QaryFalse}}\label{apx:proof_false}
We first provide a proof to Theorem~\ref{Thm_one_QaryFalse}.
\begin{IEEEproof}
We assume that whenever a tie happens, every competing hypothesis becomes acceptable simultaneously at the fusion center. This would only make the mean time to a false alarm or a false isolation smaller; hence, is valid for deriving lower bounds.

Let $q_a$ be the actual hypothesis index and
\begin{equation}
    q^*= \arg \min_{\hat{q}\in [Q]\setminus\{q_a\}}\mathbb{E}_0^{q_a,\emptyset}[\sigma^{\hat{q}}_{(d-M)}(h)].
\end{equation}
Recall that $\mathsf{P}^{q_a,\emptyset}_0$ is the probability measure when the change of type $H_{q_a}$ happens at $\nu=0$ and the compromised sensors are absent. With a fixed $d$, we have
\begin{equation}\label{eq_layer cake expectation}
\mathbb{E}_0^{q_a,\emptyset}[\tilde{\sigma}^{q^*}_{(d-M)}(h)]=\sum^\infty_{t=0}\mathsf{P}^{q_a,\emptyset}_0\left(\tilde{\sigma}^{q^*}_{(d-M)}(h)>t\right).
\end{equation}
Now, let $\cN_q(s)\triangleq\{k\in \cN:\sigma^{q}_k(h)\leq s\}$ be the set of honest sensor indices with $\sigma^{q}_{k}(h) \leq s$. For every $t\in\N$, the event $\sigma^{q^*}_{(d-M)}(h) \leq t$ happens if and only if the following is true, 
\begin{equation}
    \bigcup^t_{s=1}\left\{\left ( \bigcap^{Q}_{q\neq q^*} |\cN_q(s)|<d-M \right ) \bigcap \big \{ |\cN_{q^*}(s)| \geq d-M \big \}\right\}.
\end{equation}
Then, we have
\begin{align}
\mathsf{P}^{q_a,\emptyset}_0 \left( \sigma^{q^*}_{(d-M)}(h) \leq t \right) & \leq \mathsf{P}^{q_a,\emptyset}_0 \left( \bigcup^t_{s=1}  |\cN_{q^*}(s)| \geq d-M \right) \notag \\ & = \mathsf{P}^{q_a,\emptyset}_0 \left( |\cN_{q^*}(t)| \geq d-M \right) \label{eq_QaryFalse_dsensor1}
\end{align}
Also, we know that $|\cN_{q^*}(t)| \geq d-M$ happens if and only if there are sensor indices $k_1,\ldots, k_{d-M} \in \cN$ with $\sigma^{q^*}_{k_j}(h) \leq t$ for $j \in [d-M]$. We further bound \eqref{eq_QaryFalse_dsensor1} by union bound as follows,
\begin{align}
  \mathsf{P}^{q_a,\emptyset}_0 \left( |\cN_{q^*}(t)| \geq d-M \right) & \leq \sum_{k_1,\ldots, k_{d-M} \in \cN} \mathsf{P}^{q_a,\emptyset}_0 \left( \bigcap_{j=1}^{d-M}  \sigma^{q^*}_{k_j}(h) \leq t \right) \notag \\
   & = \sum_{k_1,\ldots, k_{d-M} \in \cN} \prod_{j=1}^{d-M} \mathsf{P}^{q_a,\emptyset}_0 \left( \sigma^{q^*}_{k_j}(h) \leq t \right) \notag \\
   & = \left(\begin{array}{c}
|\cN| \\ d-M \end{array} \right)\Big ( \mathsf{P}^{q_a,\emptyset}_0 \left( \sigma^{q^*}_{1}(h) \leq t \right) \Big)^{d-M}. \label{eq_QaryFalse_dsensor2}
\end{align}
where the first and second equalities are from the independent and identical distributions of different sensor observations, respectively.

Note that from the definition of matrix CUSUM in \eqref{eq_MatrixCUSUMacceptabletime}, it follows that
\begin{align}
  \mathsf{P}^{q_a,\emptyset}_0 \left( \sigma^{q^*}_{1}(h) \leq t \right) & =\mathsf{P}^{q_a,\emptyset}_0 \left( \bigcup^t_{s=1} \left \{ \left \{ Y^1_{s,q^*} \geq h,  q^*=\arg \max_{q' \in [Q]} Y^1_{t,q'} \right\}\bigcap^{s-1}_{s'=1} \max_{q' \in [Q]} Y^1_{s',q'}<h \right \} \right) \notag \\
   & \leq \sum_{s=1}^{t} \mathsf{P}^{q_a,\emptyset}_0 \left ( \left \{ Y^1_{s,q^*} \geq h,  q^*=\arg \max_{q' \in [Q]} Y^1_{t,q'} \right\}\bigcap^{s-1}_{s'=1} \max_{q' \in [Q]} Y^1_{s',q'}<h \right )  \notag \\
   & \leq \sum_{s=1}^{t} \mathsf{P}^{q_a,\emptyset}_0 \left ( Y^1_{s,q^*} \geq h\right ) =\sum_{s=1}^{t}\mathsf{P}^{q_a,\emptyset}_0 \left ( \bigcap_{0 \leq j \leq Q, j \neq q^*} Y^1_s(q^*,j) \geq h \right ) \notag \\
   & \leq \sum_{s=1}^{t} \mathsf{P}^{q_a,\emptyset}_0 \left (Y^1_s(q^*,0) \geq h \right ). \label{eq_QaryFalse_fsensor2}
\end{align}
Now, we know from \cite{bounds} that $\mathsf{P}^{q_a,\emptyset}_0 \left (Y^1_s(q^*,0) \geq h \right ) \leq e^{-h}$. Thus, 
from \eqref{eq_QaryFalse_dsensor2} and \eqref{eq_QaryFalse_fsensor2}, we have
\begin{equation}\label{eqn:prob_false_iso}
\mathsf{P}^{q_a,\emptyset}_0 \left( \sigma^{q^*}_{(d-M)}(h) \leq t \right) \leq \left(\begin{array}{c}
|\cN| \\ d-M \end{array} \right) t^{d-M}e^{-(d-M)h}.
\end{equation}

Plugging \eqref{eqn:prob_false_iso} into \eqref{eq_layer cake expectation} results in
\begin{align}
  \mathbb{E}_0^{q_a,\emptyset}[\sigma^{q^*}_{(d-M)}(h)] & > \sum^\infty_{t=0} \left(1-\left(\begin{array}{c}
|\cN| \\ d-M \end{array} \right) t^{d-M}e^{-(d-M)h} \right)^+  \notag \\
   & \geq \int_{0}^{\infty} \left(1-\left(\begin{array}{c}
|\cN| \\ d-M \end{array} \right) t^{d-M}e^{-(d-M)h} \right)^+dt, \label{eq_QaryFalse_fsensor3}
\end{align}
where the second inequality comes from the non-increasing property in $t$ of
\begin{equation}
1-\left(\begin{array}{c}
|\cN| \\ d-M \end{array} \right) t^{d-M}e^{-(d-M)h}.
\end{equation}

Finally, noticing that the lower bound in \eqref{eq_QaryFalse_fsensor3} is not a function of the actual hypothesis $q_a$ concludes the proof of $\A[\tau^m_{(d)}(h)]$. 
\end{IEEEproof}

In what follows, we present a proof to Theorem~\ref{CoroSimu_QaryFalse}.
\begin{IEEEproof}
Again, let $q_a$ be the actual hypothesis and let
\begin{equation}
    q^*=\arg \min_{\hat{q}\in[Q]\setminus\{q_a\}}\mathbb{E}_0^{q_a,\emptyset}[S^{\hat{q}}_{d-M}(h)].
\end{equation}
With a fixed $d$, we have
\begin{equation}\label{eq_simu_layer cake expectation}
\mathbb{E}_0^{q_a,\emptyset}[S^{\;q^*}_{(d-M)}(h)]=\sum^\infty_{t=0}\mathsf{P}^{q_a,\emptyset}_0\left(S^{\;q^*}_{(d-M)}(h)>t\right).
\end{equation}
Note that for every $t\in\N$, the event $S^{\;q^*}_{(d-M)}(h) \leq t$ happens if and only if the following event is true,
\begin{equation}
    \bigcup^t_{s=1}\left\{\left ( \bigcap^{Q}_{q\neq q^*} Y^{(K-(d-M)+1)}_{s,q} < h \right ) \bigcap Y^{(K-(d-M)+1)}_{s,q^*} \geq h \right\}.
\end{equation}
Then, we have
\begin{align}
\mathsf{P}^{q_a,\emptyset}_0 \left( S^{\;q^*}_{(d-M)}(h) \leq t
\right) & \leq \mathsf{P}^{q_a,\emptyset}_0 \left(
\bigcup^t_{s=1}  Y^{(K-(d-M)+1)}_{s,q^*} \geq h  \right) \notag \\
& \leq \sum_{s=1}^{t} \; \mathsf{P}^{q_a,\emptyset}_0 \left(
Y^{(K-(d-M)+1)}_{s,q^*} \geq h \right)
\label{eq_simuQaryFalse_dsensor1}
\end{align}
Also, we know that event $Y^{(K-(d-M)+1)}_{s,q^*} \geq h$ happens if and only if there are $d-M$ sensors with indexes $k_1,\ldots, k_{d-M} \in \cN$ which have local decisions $q^*$ at time index $s$. Therefore,
\begin{align}
  \mathsf{P}^{q_a,\emptyset}_0 \left( Y^{(K-(d-M)+1)}_{s,q^*} \geq h \right)  &= \sum_{k_1,\ldots, k_{d-M} \in \cN} \prod_{j=1}^{d-M} \mathsf{P}^{\emptyset}_{q_a} \left ( Y^{k_j}_{s,q^*} \geq h  \right ) \notag \\
   & = \left(\begin{array}{c}
|\cN| \\ d-M \end{array} \right)\Big ( \mathsf{P}^{q_a,\emptyset}_0 \left(  Y^{1}_{s,q^*} \geq h \right) \Big)^{d-M} \label{eq_simuQaryFalse_dsensor2}
\end{align}
where the first and second equalities are from the independent and identical distributions of different sensor observations, respectively. Now as in \eqref{eq_QaryFalse_fsensor2}, it follows that
\begin{equation}\label{eq_simuQaryFalse_fsensor3}
  \mathsf{P}^{q_a,\emptyset}_0 \left(  Y^{1}_{s,q^*} \geq h \right) \leq \mathsf{P}^{q_a,\emptyset}_0 \left (Y^1_s(q^*,0) \geq h \right ) \leq e^{-h}
\end{equation}
Thus, from \eqref{eq_simuQaryFalse_dsensor1}-\eqref{eq_simuQaryFalse_fsensor3},
\begin{equation}\label{eq_simuQaryFalse_fsensor4}
\mathsf{P}^{q_a,\emptyset}_0 \left( S^{\;q^*}_{(d-M)}(h) \leq t
\right) \leq \left(\begin{array}{c} |\cN| \\ d-M \end{array}
\right) t e^{-(d-M)h}.
\end{equation}
Plugging \eqref{eq_simuQaryFalse_fsensor4} into \eqref{eq_simu_layer cake expectation} and noticing that the bound in \eqref{eq_simuQaryFalse_fsensor4} is independent of $q_a$ completes the proof for the lower bound on $\A(\tau^s_{(d)}(h))$. 
\end{IEEEproof}


\end{document}